\pdfoutput=1
\documentclass[journal,10pt]{IEEEtran}
\usepackage{cite}
\usepackage{indentfirst}
\usepackage{graphicx}
\usepackage{changepage}
\usepackage{stfloats}
\usepackage{amsfonts,amssymb}
\usepackage{bm}
\usepackage{algorithm}
\usepackage{algorithmic}
\usepackage{amsmath}
\usepackage{amsmath,amssymb,amsfonts}
\usepackage{times}
\usepackage{url}
\usepackage{cite}
\usepackage{textcomp}
\usepackage{xcolor}
\usepackage{color}
\usepackage{setspace}
\usepackage{enumitem}
\usepackage{float}
\usepackage{diagbox}
\usepackage[T1]{fontenc}
\usepackage[utf8]{inputenc}
\usepackage{authblk}
\usepackage{threeparttable}
\usepackage{booktabs}
\usepackage{footnote}
\usepackage{graphicx}
\usepackage{pifont}
\usepackage{subfigure}
\usepackage{mathrsfs}
\allowdisplaybreaks[4]

\newenvironment{proof}{{\indent  \indent \it Proof:}}{\hfill $\blacksquare$}

\begin{document}
\title{Throughput Maximization for UAV-enabled Integrated Periodic Sensing and Communication}

\newtheorem{thm}{\bf Lemma}
\newtheorem{remark}{\bf Remark}
\newtheorem{Pro}{\bf Proposition} 
\author{
	Kaitao Meng, \textit{Member, IEEE}, Qingqing Wu, \textit{Senior Member, IEEE}, Shaodan Ma, \textit{Senior Member, IEEE}, Wen Chen, \textit{Senior Member, IEEE}, Kunlun Wang, \textit{Member, IEEE}, and Jun Li, \textit{Senior Member, IEEE}
	\thanks{This work was supported in part by the FDCT under Grant 0119/2020/A3, GDST under Grant 2020B1212030003, STDF Macau SAR  No. 0036/2019/A1, the National key project 2020YFB1807700 and 2018YFB1801102, and NSFC Grant No. 61872184.}
	\thanks{{K. Meng, Q. Wu, and S. Ma are with the State Key Laboratory of Internet of Things for Smart City, University of Macau, Macau, 999078, China. }{(emails: \{kaitaomeng, qingqingwu, shaodanma\}@um.edu.mo).} W. Chen is with the Department of Electronic Engineering, Shanghai Jiao Tong University, Shanghai 201210, China (email: wenchen@sjtu.edu.cn). K. Wang is with the School of Communication and Electronic Engineering, East China Normal University, Shanghai 200241, China (email: klwang@cee.ecnu.edu.cn). J. Li is with the School of Electronic and Optical Engineering, Nanjing University of Science Technology, Nanjing 210094, China (email: jun.li@njust.edu.cn). }
}

\maketitle


\begin{abstract}
	Driven by unmanned aerial vehicle (UAV)'s advantages of flexible observation and enhanced communication capability, it is expected to revolutionize the existing integrated sensing and communication (ISAC) system and promise a more flexible joint design. Nevertheless, the existing works on ISAC mainly focus on exploring the performance of both functionalities simultaneously during the entire considered period, which may ignore the practical asymmetric sensing and communication requirements. In particular, always forcing sensing along with communication may make it is harder to balance between these two functionalities due to shared spectrum resources and limited transmit power. To address this issue, we propose a new integrated periodic sensing and communication (IPSAC) mechanism for the UAV-enabled ISAC system to provide a more flexible trade-off between two integrated functionalities. Specifically, the system achievable rate is maximized via jointly optimizing UAV trajectory, user association, target sensing selection, and transmit beamforming, while meeting the sensing frequency and beam pattern gain requirement for the given targets. Despite that this problem is highly non-convex and involves closely coupled integer variables, we derive the closed-form optimal beamforming vector to dramatically reduce the complexity of beamforming design, and present a tight lower bound of the achievable rate to facilitate UAV trajectory design. Based on the above results, we propose a two-layer penalty-based algorithm to efficiently solve the considered problem. To draw more important insights, the optimal achievable rate and the optimal UAV location are analyzed under a special case of infinity number of antennas. Furthermore, we prove the structural symmetry between the optimal solutions in different ISAC frames without location constraints in our considered UAV-enabled ISAC system. Based on this, we propose an efficient algorithm for solving the problem with location constraints. Numerical results validate the effectiveness of our proposed designs and also unveil a more flexible trade-off in ISAC systems over benchmark schemes.
\end{abstract}

\section{Introduction}
\par
Driven by spectrum reuse potential and enormous demands of robust sensing ability, there is a recent surge of interest in the development of integrated (radar) sensing and communications (ISAC) techniques for both academia and industry \cite{Hassanien2019Dual}, \cite{Ericsson2020}. Different from the spectrum sharing between separate radar sensing and communication systems \cite{Li2016Optimum}, ISAC shares the same wireless infrastructures for simultaneously conveying information to the receiver and extracting information from the scattered echoes \cite{Yuan2021Integrated}. Thus, ISAC could not only achieve integration gain to significantly enhance the spectrum utilization efficiency and reduce hardware costs, but also introduce coordination gain to efficiently balance between two functionalities' performance \cite{Yuan2021Integrated, Zhang2021Overview}. With the advancements of massive antennas and millimeter wave (mmWave)/terahertz (THz), ISAC base stations (BSs) could also provide higher sensing resolution and accuracy to enable many location-aware intelligent applications with stringent sensing requirements \cite{Godrich2010Target}. Several similar terminologies have been utilized to describe this related research, such as radar-communication (RadCom) \cite{Sturm2011Waveform, Oliveira2021Joint}, dual-functional radar communication (DFRC) \cite{Liu2018Toward, Liu2021Dual}, joint communication and radar sensing (JCAS) \cite{Mishra2019Toward, Zhang2019Multibeam}. In the industry, ISAC is regarded as a key technology in Huawei and Nokia for future wireless network investigations \cite{Pin2021Integrated, Wild2021Joint}; "Hexa-X" project supported by European Commission focuses on extending the localization and sensing capabilities for 6G \cite{Wymeersch2021Integration}; Project IEEE 802.11bf plans to develop WLAN sensing by analyzing the received WLAN signals to recognize the features of the intended targets in a given environment \cite{WLANSENSING}.
\par
The prior works on ISAC systems have shown that co-designed waveform and beamforming could provide mutual benefits of both sensing and communication \cite{Zhang2021Design, Cui2021Integrating, Liu2020JointRadar, Chen2021Radar, Liu2018MIMO}. For instance, a joint transmit beamforming model was proposed to optimize the radar transmit beam pattern while meeting the requirement of the signal-to-interference-plus-noise ratio (SINR) at each communication user \cite{Liu2020JointRadar}. The authors in \cite{Chen2021Radar} proposed a Pareto optimization framework of the DFRC system to analyze the achievable performance region of communication and sensing. However, the performance of sensing is generally dependent on the explicit line-of-sight (LoS) links between targets and transceivers, while non-Los (NLoS) links are treated as unfavorable interference for the target sensing. For the potential targets located far away from BSs or blocked by obstacles, the sensing performance will severely degrade or the sensing missions may even fail because of serious path loss of the echoed signals. Hence, terrestrial ISAC BSs could only provide sensing and communication services within a fixed range due to limited transmit power and NLoS signal paths caused by surrounding obstacles.
\par
Driven by the unmanned aerial vehicle (UAV)' on-demand deployment and strong LoS links features \cite{Gupta2016Survey, Hua2019Energy}, it is expected to be a cost-effective aerial platform to provide enhanced ISAC service. In particular, more flexible observation, better communication quality, larger service coverage could be achieved by exploiting the high mobility of UAVs \cite{Wu2019Fundamental, Hua2018Power}. Traditional works on UAV-enabled wireless networks mainly focused on the separate design of sensing and/or communication \cite{Meng2021Space, Zhang2020Age, Meng2021Sensing}, instead of considering integrated waveform and beamforming design for sensing and communication. Different from the separate-design sensing and communication systems, the achievable rate for the UAV-enabled ISAC system is influenced by multiple complicated factors, including beam pattern constraints, resource allocation, as well as beamforming design closely coupled with UAV trajectory. Therefore, this difference leads to a new challenge for the achievable rate maximization problem in UAV-enabled ISAC systems. Most recently, there are several works studying the trajectory or deployment optimization issue in UAV-based ISAC \cite{lyu2021joint, Chen2020Performance, wang2021qos}. For instance, the authors in \cite{lyu2021joint} proposed a joint UAV maneuver and transmit beamforming optimization algorithm to maximize the communication performance while ensuring the sensing requirements for the given targets. By deploying multiple UAVs to perform tasks cooperatively, greater coverage of ISAC networks can be achieved \cite{Chen2020Performance}. Besides, ISAC-enabled cellular networks can be utilized to monitor and localize the suspicious UAV targets in the sky to protect the physical security \cite{wei2021safeguarding}. 
\par
However, the above works on ISAC \cite{Zhang2021Design, Cui2021Integrating, Liu2020JointRadar, Chen2021Radar, Liu2018MIMO, lyu2021joint, Chen2020Performance, wei2021safeguarding} mainly focused on exploring the performance of both functionalities simultaneously during the entire considered period, where all sensing tasks are performed together with communication all the time. This may ignore the asymmetric sensing and communication requirements in practical systems. In other words, the sensing frequency could be different from the data frame rate. For example, for target tracking scenarios, a relatively low/high sensing frequency is preferred for a low-speed/high-speed object. Hence, sensing frequency should be set based on the targets' motion state and the timeliness requirement of the specific tasks. Nonetheless, this important aspect of ISAC systems, sensing frequency, has not been taken into account in the literature. On the other hand, always forcing sensing along with communication all the time may introduce excessive sensing, making it is harder to balance between these two functionalities. Furthermore, excessive sensing may result in the waste of spectrum resources and stronger interference to communication users, thereby limiting the performance of communication users. Moreover, forcing both functionalities to work simultaneously will also inevitably cause higher energy consumption, which is unfavorable for the equipment with insufficient energy (e.g., power limited UAVs \cite{Wu2021Comprehensive}). Therefore, there is an urgent need to investigate the achievable rate improvement in such scenarios by considering the sensing frequency besides the commonly used sensing power, especially for UAV-enabled ISAC systems due to its autonomous mobility. Note that the fixed-deployment ISAC system considering the sensing frequency is actually a special case of our work. By optimizing the UAV trajectory, the flexibility of beam design and the efficiency of task association for ISAC systems can be further improved. This knowledge gap motivates us to develop effective UAV-enabled ISAC mechanisms to fulfill a more general and flexible trade-off between sensing and communication.
\par
With the above consideration, we study a UAV-enabled ISAC system where one UAV is dispatched to perform sensing tasks while providing downlink communication services for several single-antenna users, as shown in Fig.~{\ref{figure1}}. Considering the practical sensing frequency requirements, we propose an integrated periodic sensing and communication (IPSAC) mechanism where all sensing tasks are periodically executed along with the communication service. Specifically, the achievable rate maximization problem is investigated by jointly optimizing the transmit beamforming, user association, sensing time selection, and UAV trajectory in this work, subject to the sensing frequency and beam pattern gain requirements. As compared to traditional ISAC considered in \cite{lyu2021joint}, which always forces the UAV to perform sensing tasks and provide communication service at the same time, our proposed scheme is more general and offers more flexibility to balance between practical sensing and communication over time. Besides, by setting the frequency to infinity or the minimum threshold, it is not difficult to find that both standalone communication and always-sensing are special cases of our considered periodic sensing and communication scenarios. 
\par 
\begin{figure}[t]
	\centering
	\setlength{\abovecaptionskip}{0.cm}
	\includegraphics[width=8cm]{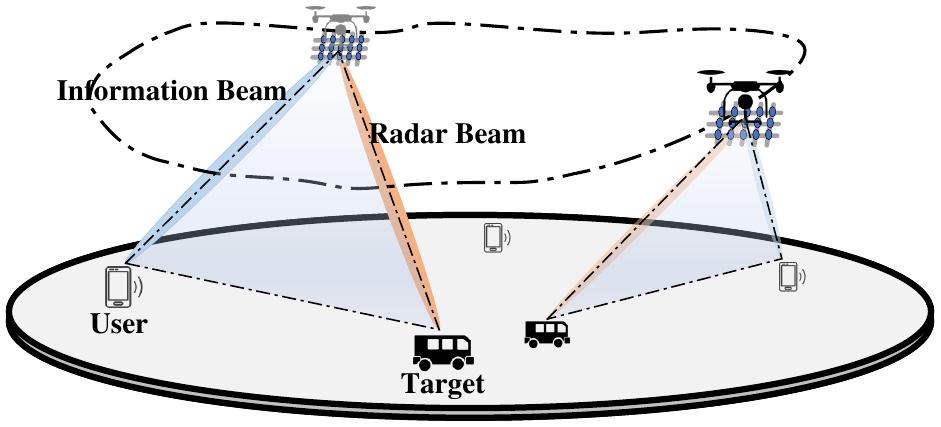}
	\caption{The illustration of UAV-enabled integrated sensing and communication scenarios.}
	\label{figure1}
\end{figure}
\par 
However, solving this periodic ISAC optimization problem is highly non-trivial. Specifically, it is non-convex and involves integer variables which are closely coupled with UAV trajectory and beamforming vectors. Unlike traditional trajectory optimization problem for single-antenna UAVs, joint beamforming and UAV trajectory optimization problem for ISAC is very complicated, since the location of the UAV is coupled with beamforming vector in a more complex form. Also, the complexity of the trajectory discretization-based method will become intractable in practical scenarios with long mission periods \cite{Wu2018Common}. To address this issue, we first propose a two-layer penalty-based algorithm to solve the achievable rate maximization problem by decoupling the optimized variables and then propose a low-complexity algorithm to solve the considered problem more efficiently. The main contribution in this paper is summarized as follow: 
\begin{itemize}
	\item First, we propose a UAV-enabled IPSAC mechanism to achieve a more general and flexible trade-off between sensing power requirement, sensing frequency, and communication performance for multi-users and multi-targets scenarios. Furthermore, we formulate a periodic ISAC problem to maximize the achievable rate while satisfying sensing frequency and beam pattern gain constraints.
	\item Next, we derive the closed-form beamforming vector under any given UAV location, and present the closed-form optimal achievable rate and sensing location if the number of antennas is infinity, thereby providing guidance for algorithm design. By introducing a tight lower bound of the original objective function, a penalty-based algorithm is proposed to jointly optimize beamforming, user association, sensing time selection, and UAV trajectory.
	\item Furthermore, to draw useful insights, we prove a novel characteristic of structural symmetry between the optimal solutions in different ISAC frames without initial and final location constraints. Accordingly, we reveal the monotonic relationship between sensing frequency and communication capacity in our considered IPSAC system. Based on this, a low-complexity solution can be constructed while achieving high-quality performance. 
	\item Finally, simulation results unveil a more flexible trade-off in ISAC systems over benchmark schemes and show that the UAV trajectory design plays an important role in balancing sensing and communication performance in IPSAC mechanisms. It is also found that the UAV tends to provide communication services while sensing the target closer to the associated user.
\end{itemize}

The remainder of this paper is organized as follows. Section \ref{SYSTEM} introduces the system model and problem formulation of the UAV-enabled IPSAC system. In Section \ref{AnalysisOfOptimalSolution}, we derive the closed-form optimal beamforming vector and propose a penalty-based algorithm to address the sum achievable rate maximization problem. Section \ref{WithoutLocationConstraint} presents the symmetrical structure characteristic among ISAC frames and a low-complexity algorithm. Section \ref{Simulations} provides numerical results to validate the performance of our proposed mechanism. Section \ref{Conclusion} concludes this paper.

\textit{Notations}: $\|{\bm{x}} \|$ denotes the Euclidean norm of a complex-valued vector ${\bm{x}}$. For a general matrix ${\bm{X}}$, $\operatorname{rank}({\bm{X}})$, ${\bm{X}}^H$, ${\bm{X}}^T$, and $[{\bm{X}}]_{p,q}$ denote its rank, conjugate transpose, transpose, and the element in the $p$th row and $q$th column, respectively. For a square matrix ${\bm{Y}}$, ${\rm{tr}}({\bm{Y}})$ and ${\bm{Y}}^{-1}$ denotes its trace and inverse, respectively, while ${\bm{Y}} \succeq 0$ represents that ${\bm{Y}}$ is a positive semidefinite matrix. $\jmath$ denotes the imaginary unit, i.e., ${\jmath}^2 = -1$. The distribution of a circularly symmetric complex Gaussian (CSCG) random variable with mean $x$ and variance $\sigma^2$ is denoted by $\mathcal{C} \mathcal{N}(x,\sigma^2)$.

\section{System Model and Problem Formulation}
\label{SYSTEM}
We consider a UAV-enabled ISAC system aimed at sensing several prospective ground targets while providing downlink communication service for $K$ single-antenna users within a given flight period $T$ s. The set of the users and that of the prospective targets are denoted by ${\cal{K}} = \{1,\cdots,K\}$ and ${\cal{J}} = \{1,\cdots,J\}$, respectively. The horizontal location of user $k$ is denoted by ${\bm{u}}_k = [u_{x,k}, u_{y,k}]^T$, which can be either obtained by global positioning system (GPS) or estimated by uplink signals \cite{Garcia2017Direct}. The horizontal locations of the potential targets are denoted by ${\bm{v}}_j = [v_{x,j}, v_{y,j}]^T$, $j \in {\cal{J}}$. The value of ${\bm{v}}_j$ is determined based on the specific sensing tasks. For example, ${\bm{v}}_j$ can be set as the estimated location based on the previous frames for target tracking, or set as a uniformly sampled positions in the region of interest for target detection. The whole mission period $T$ can be discretized into $N$ time slots with duration $\delta_t = \frac{T}{N}$, and the index of time slot is denoted by $n \in {\cal{N}} =  \{1,\cdots,N\}$. Here, the time slot is chosen to be sufficiently small, during which the UAV's location is assumed to be approximately unchanged to facilitate the trajectory and beamforming design for ISAC. The UAV's horizontal location is denoted by ${\bm{q}}[n] = [ q_x[n], q_y[n]]^T$, where $n \in {\cal{N}}$, and the UAV is assumed to fly at a constant altitude of $H$ m subject to air traffic control \cite{Wu2018Joint}. The general uniform plane array (UPA) is adopted at the UAV, where the number of antennas is denoted by $M = M_x \times M_y$ with $M_x$ and $M_y$ denoting the number of elements along the $x$- and $y$-axis, respectively. The adjacent elements are separated by $d_x = d_y = \frac{\lambda}{2}$, where $\lambda$ denotes the carrier wavelength. Specifically, the UPA is parallel to the ground to facilitate the technical derivation, as shown in Fig.~\ref{figure1}. 

\begin{figure}[t]
	\centering
	\setlength{\abovecaptionskip}{0.cm}
	\includegraphics[width=8.2cm]{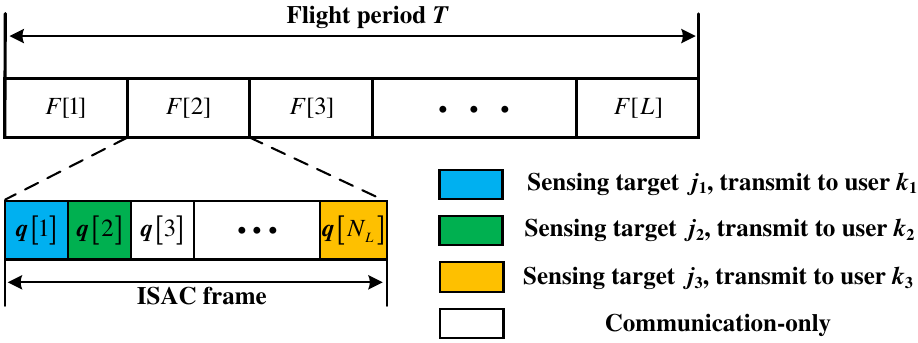}
	\caption{IPSAC mechanism for multi-users and multi-targets scenarios. }
	\label{figure2}
\end{figure}
\subsection{ISAC Frame}
Based on the practical timeliness requirements of sensing tasks, we propose an IPSAC mechanism for multi-user and multi-target scenarios to find a fundamental trade-off between sensing and communication. Specifically, it is assumed that each sensing task should be performed at least once in each ISAC frame, as shown in Fig.~{\ref{figure2}}. Assuming that the total frame number $L = \frac{T}{T_L}$ is an integer for ease of analysis, where $T_L$ is the frame length.\footnote{The length of the ISAC mission frame $T_L$ is set according to the requirement of the task execution frequency.} Then, the number of time slots in each ISAC frame is $N_L  = \frac{N}{L}$ and the index of ISAC frame is denoted by $l \in {\cal{L}}= \{1,\cdots,L\}$. In our proposed IPSAC mechanism, time division multiple access (TDMA) is adopted to avoid signal interference between different information beams due to strong LoS channel, while each target could be sensed in any time slot of each ISAC frame. If the UAV aims to sense target $j$ at time slot $n$, we denote $c_j[n] = 1$. Otherwise, $c_j[n] = 0$. Also, at most one target can be sensed in each time slot. By performing sensing tasks separately in different time slots, the computational complexity of the target estimation algorithm can be reduced. Based on the above discussion, the following conditions hold:
\begin{equation}\label{SensingFrequencyConstraintsA}
	\sum\nolimits_{n = (l - 1){N_L} + 1}^{l  {N_L}} {{c_j}[n]} = 1, \forall l, j,
\end{equation}
\begin{equation}\label{SensingFrequencyConstraintsB}
	\sum\nolimits_{j = 1}^{J} {{c_j}[n]}  \le 1, \forall n.
\end{equation}
Then, the sensing frequency of each target is defined as $1 / T_L = 1 / (\delta N_L)$.
\subsection{Communication and Sensing Model}
\par
The communication links between the UAV and the user are assumed to be dominated by the LoS component \cite{Wu2018Capacity}. Hence, the aerial-ground channel follows the free-space path loss model and the channel power gain from the UAV to user $k$ can be expressed as
\begin{equation}
	\beta_k({\bm{q}}[n], {\bm{u}}_k) \!=\! \beta_{0} d({\bm{q}[n], {\bm{u}}_k})^{-2} \!=\! \frac{\beta_{0}}{H^{2}\!+\! \|{\bm{q}}[n] - {\bm{u}}_k\|^{2}},
\end{equation}
where $\beta_0$ represents the channel power at the reference distance 1 m. Besides, the Doppler effect induced by the UAV mobility is assumed to be well compensated at the communication users \cite{Wu2019Fundamental, kang2010fine} and the sensing receiver \cite{Xing2009Motion, pieraccini2019doppler}, respectively. The transmit array response vector of the UAV towards user $k$'s location ${{\bm{u}}_k}$ is
\begin{equation}\label{SteeringVector}
\begin{aligned}
	{\bm{a}}^H(\bm{q}[n], {\bm{u}}_k) 
	=& \left[1, \cdots, e^{ -\frac{\jmath 2 \pi(M_{x}-1) {d}_{x} {\Phi}(\bm{q}[n], {\bm{u}}_k)}{\lambda}}\right] \\ 
	&\otimes  \left[1, \cdots, e^{ -\frac{ \jmath 2 \pi(M_{y}-1) {d}_{y} {\Omega}(\bm{q}[n], {\bm{u}}_k)}{\lambda}}\right].
\end{aligned}
\end{equation}
In (\ref{SteeringVector}), $ {\Phi}(\bm{q}[n], {\bm{u}}_k) = \sin(\phi({\bm{q}}[n],{\bm{u}}_k)) \cos(\theta({\bm{q}}[n],{\bm{u}}_k)) = \frac{q_x[n] - u_{x,k}}{\|\bar{\bm{q}}[n] - \bar{\bm{u}}_k\|}$, where $\bar{\bm{q}}[n] = [q_x, q_y, H]^T$, and $\bar{\bm{u}}_k = [u_{x,k},u_{y,k},0]^T$. And ${\Omega}(\bm{q}[n], {\bm{u}}_k) =  \sin(\phi({\bm{q}}[n],{\bm{u}}_k)) \sin(\theta({\bm{q}}[n],{\bm{u}}_k)) = \frac{q_y[n] - u_{y,k}}{\|\bar{\bm{q}}[n] - \bar{\bm{u}}_k\|}$. $\phi({\bm{q}}[n],{\bm{u}}_k)$ represents the zenith angle of departure (AoD) of the signal from the UAV to user $k$'s location ${\bm{u}}_k$, and $\theta({\bm{q}}[n],{\bm{u}}_k)$ represents its corresponding azimuth AoD. Therefore, the baseband equivalent channel from the UAV to user $k$ can be expressed as 
\begin{equation}
	{\bm{h}}^H_k\left(\bm{q}[n], {\bm{u}}_k\right)=\sqrt{\beta_{k}(\bm{q}[n],{\bm{u}}_k)} e^{-\jmath \frac{2 \pi d{(\bm{q}[n],{\bm{u}}_k)}}{\lambda}} {\bm{a}}^H\left(\bm{q}[n], {\bm{u}}_k\right) .
\end{equation}
\par
Without loss of generality, we assume that the UAV can transmit the information-bearing signal ${\bm{s}}_k$ to user $k$, where ${s}_k, \sim \mathcal{C} \mathcal{N}(0,1)$. Moreover, the communication signals are uncorrelated with each other, i.e., ${\bm{E}}\left(s_k[n]s_{k'}[n]\right) = 0$, where $k \ne k'$, and $k, k' \in {\cal{K}}$ \cite{LiuX2020Joint}. The linear transmit precoding is applied at the UAV for the assigned user and target. Hence, the complex baseband transmitted signal at the UAV can be expressed as a weighted sum of communication signals, i.e., 
\begin{equation}
	{\bm{x}}[n] = {{\bm{w}}_c}[n] \sum\nolimits_{k=1}^{K}{\alpha_k[n] {s_k}[n]}, n \in {\cal{N}},
\end{equation}
where ${\bm{w}}_c[n] \in \mathbb{C}^{M \times 1}$ is the corresponding information beamforming vector, and $\alpha_k[n] = 1$ if the UAV transmits signal $s_k$ to user $k$ at the $n$th time slot, otherwise, $\alpha_k[n] = 0$. Since the UAV only serves at most one user at each time slot, we have the following constraint
\begin{equation}\label{UserAssociation}
	\sum\nolimits_{k=1}^{K} \alpha_k[n] \le 1, \forall, n.
\end{equation}
Then, at the $n$th time slot, the received signal at user $k$ is
\begin{equation}\label{UserReceivedSignal}
	{y}_k[n] =  {\bm{h}}^H_{c,k}\left(\bm{q}[n], {\bm{u}}_k\right) ({{\bm{w}}_c}[n] \sum\nolimits_{k=1}^{K}{\alpha_k[n]  {s_k}[n]} ) + {n_k}[n],
\end{equation}
where $n_k[n] \sim {\cal{CN}}(0, \sigma_k^2)$ denotes the additive white Gaussian noise (AWGN) at user $k$'s receiver. Accordingly, for $ \alpha_k[n] = 1$, the signal-to-noise ratio (SNR) of user $k$ is given by 
\begin{equation}
	{\gamma}_k[n] = \frac{{{{\left| {{\bm{h}}^H_{c,k}\left(\bm{q}[n], {\bm{u}}_k\right) {{\bm{w}}_c}[n]} \right|}^2}}}{  \sigma_k^2} ,\forall n \in {\cal{N}}.
\end{equation}
As a result, when $\alpha_k[n] = 1$, the corresponding achievable rate of user $k$ at time slot $n$ in bits-per-second-per-Hertz ({\rm{bps/Hz}}) is
\begin{equation}
	{R}_k[n] = {\log_2 (1 + {\gamma_k }[n])}.
\end{equation}
As the communication signals reflected by the target can also be utilized for target parameter estimation in our considered system \cite{Cui2021Integrating, Wang2019Dual}, the communication signals $\{s_k[n]\}_{k=1}^K$ are further exploited for sensing. As a result, the transmit beam pattern gain from the UAV to the direction of target $j$ can be given by
\begin{equation}\label{BeampatternGain}
	\begin{aligned}
		& \Gamma \left(\bm{q}[n], {\bm{v}}_j\right) = E\left[ {{{\left| {\bm{a}}^H(\bm{q}[n], {\bm{v}}_j) \left( {\bm{x}}[n]\right) \right|}^2}} \right] \\
		= & {\bm{a}}^H(\bm{q}[n], {\bm{v}}_j)  \underbrace{   \left( {\bm{w}}_c[n]{\bm{w}}^H_c[n]  \right) }_{\text{covariance matrix}}   {\bm{a}}(\bm{q}[n], {\bm{v}}_j).
	\end{aligned}
\end{equation}
Based on the definition in (\ref{BeampatternGain}), the power of reflected signals from target can be expressed a function of $\Gamma \left(\bm{q}[n], {\bm{v}}_j\right)$ together with pathloss from the UAV to the given target, as shown in constraints (\ref{P1}a).\footnote{In the proposed sensing scheme, the interference from multiple targets is practically weak and thus is approximately ignored, since the targets are sensed in a time-division multiplexing (TDM) manner along with communication and the beam power is mainly concentrated in the direction of the intended target. Moreover, the negative effects of clutter can be mitigated through the prior knowledge of clutter and precoding techniques \cite{Aubry2013Knowledge}.}
\subsection{Problem Formulation}
In this paper, we aim to maximize the achievable rate by optimizing the beamforming vector, user association, sensing time selection, and UAV trajectory, subject to the requirements of the sensing frequency, sensing power, and quality of service (QoS). Accordingly, the optimization problem is formulated as \footnote{The constraints in ({\ref{P1}a}) can be extended into the case with the pathloss exponent of 4, i.e., $ {{c_j}[n]}  \frac{\Gamma \left(\bm{q}[n], \bm{v}_j\right)}{d({\bm{q}}[n], {\bm{v}}_j)^4}  \ge {{c_j}[n]}  {\Gamma ^{th}_j}$, representing that the SNR of the reflected signal from targets should be larger than a given threshold. Such scenario with 4-exponent pathloss can be deemed as mono-static sensing system, while the constraint (11a) with 2-exponent pathloss represents the signal power at the location of targets should be larger than the threshold, which can be regarded as bi-static sensing scenarios, i.e. there exists another dedicated receiver for echoes analysis. Also, the comparison for these two cases are presented in simulation results, given in Section \ref{DifferentPathloss}.}
\begin{alignat}{2}
	\label{P1}
	(\rm{P1}): \quad & \begin{array}{*{20}{c}}
		\mathop {\max }\limits_{{{\bm{w}}_c}, {\bm{A}}, {{\bm{Q}}}, {\bm{C}}} \quad \mathop  \frac{1}{N}  \sum\nolimits_{n = 1}^N \sum\nolimits_{k = 1}^K \alpha_k[n] R_k[n]
	\end{array} & \\ 
	\mbox{s.t.}\quad
	& (\ref{SensingFrequencyConstraintsA}), (\ref{SensingFrequencyConstraintsB}), (\ref{UserAssociation}), \nonumber \\
	& {{c_j}[n]}  \frac{\Gamma \left(\bm{q}[n], \bm{v}_j\right)}{d({\bm{q}}[n], {\bm{v}}_j)^2}  \ge {{c_j}[n]}    {\Gamma ^{th}_j},\forall j,   n, & \tag{\ref{P1}a}\\
	& c_j[n] \in \{0,1\},  \alpha_k[n] \in \{0,1\}, \forall j,  k,  n,   & \tag{\ref{P1}b} \\
	& \frac{1}{N_L}\sum\nolimits_{n = (l - 1){N_L} + 1}^{l  {N_L}} \alpha_k[n] {R}_k[n] \ge {R}_k^{th},\forall  k, l, & \tag{\ref{P1}c} \\
	&  \left\| {\bm{w}}_c[n]\right\|^2  \le P_{\max}, \forall n,  & \tag{\ref{P1}d} \\ 
	& \| {\bm{q}}[n] -{\bm{q}}[n-1] \|  \le V_{\max} \delta_t, \forall n \in {\cal{N}} \backslash \{1\},  & \tag{\ref{P1}e} \\
	& {\bm{q}}[1] = {\bm{q}}_I,  {\bm{q}}[N] = {\bm{q}}_F.   & \tag{\ref{P1}f}
\end{alignat} 
In (P1), ${\bm{C}} = \{{\bm{c}}[n]\}_{n=1}^N$ and ${\bm{A}} = \{{\bm{\alpha}}[n]\}_{n=1}^N$, where ${\bm{c}}[n] = \{c_j[n]\}_{j=1}^J$ is the target selection at the $n$th time slot and ${\bm{\alpha}[n]} = \{{{\alpha}_{k}}[n]\}_{k=1}^K$ is the user association at the $n$th time slot. Similarly, ${{\bm{w}}_c} = \{{\bm{w}}_c[n]\}_{n=1}^N$, and ${{\bm{Q}}} = \{{{\bm{q}}}[n]\}_{n=1}^N$. Under the given sensing frequency, the beam pattern gain constraints at the direction of targets are given by (\ref{P1}a), where ${\Gamma ^{th}_j}$ denotes the beam pattern gain threshold of target $j$ and $d({\bm{q}}[n], {\bm{v}}_j)^2$ represents the corresponding pathloss. The minimum achievable rate requirements in each ISAC frame are given by (\ref{P1}c) to satisfy the quality of service. The total transmit power and the maximum distance between two consecutive locations are constrained as in (\ref{P1}d) and (\ref{P1}e), respectively. The initial and final locations constraints are given by (\ref{P1}f). Besides, if a certain target needs both communication and sensing services (e.g., the sensing results can be utilized for communication enhancement, i.e., sensing gain achieved for communication), another user with the same location could be introduced for this case.
\par
Solving problem (P1) is highly non-trivial, since it is non-convex and involves integer variables which are closely coupled with UAV trajectory and beamforming. To address this problem, we first derive the closed-form optimal beamforming vector and a tight lower bound of the achievable rate. Accordingly, an efficient penalty-based algorithm consisting of two layers is proposed to solve the considered problem. Furthermore, by ignoring initial and final location constraints, we prove the structural symmetry between the optimal solutions in different ISAC frames. Based on this result, a low-complexity algorithm is proposed to reduce the computation complexity caused by trajectory discretization, especially for the practical scenarios with long flight periods.

\section{Penalty-based Algorithm to (P1)}
\label{AnalysisOfOptimalSolution}
In this section, we first investigate the closed-form optimal beamforming vector for the proposed IPSAC mechanism in Section \ref{ClosedFormBeamforming}. Then, a tight lower bound of the original objective value is provided in Section \ref{LowerBoundRate}, based on which, we propose a penalty-based algorithm to jointly optimize the UAV trajectory, user association, and sensing time selection in Section \ref{PenaltyAlgorithm} and Section \ref{InnerLayerAlgorithm}.
\subsection{Closed-form Optimal Beamforming}
\label{ClosedFormBeamforming}
\par
It can be found that, if $ \sum\nolimits_{j=1}^{J} c_j[n] = 0$ and $\alpha_{k}[n] = 1$, for any given UAV location, the optimal beamforming vector ${\bm{w}}_c^* = {\sqrt{P_{\max}}\frac{{\bm{h}}_{c,k}(\bm{q}[n],{{\bm{u}}}_k)}{\|{\bm{h}}_{c,k}(\bm{q}[n],{{\bm{u}}}_k)\|}}$. Otherwise, if $ c_j[n] = 1$ and $\alpha_{k}[n] = 1$, the optimal beamforming vector is highly coupled with the UAV trajectory. For notation convenience, denote ${\bm{h}}^H_{c,k}(\bm{q}[n],{{\bm{u}}}_k)$ and $\frac{{\bm{a}}^H(\bm{q}[n],{\bm{v}}_j)}{d({\bm{q}}[n], {\bm{v}}_j)}$ as ${\bm{h}}^H_{c,k}$ and ${\bm{h}}^H_{r,j}$, respectively. Since maximizing $R_k[n]$ is equivalent to maximizing the corresponding received signal strength of ${{\bm{w}}_c^H} {\bm{h}}_{c,k} {\bm{h}}_{c,k}^H{\bm{w}}_c$, the $\log$ function is dropped in the objective function for simplicity. The received signal strength maximization problem is reduced to
\begin{alignat}{2}
	\label{P1.1}
	\mathop {\max }\limits_{ {\bm{w}}_c}	& \begin{array}{*{20}{c}}
		\mathop {{\bm{w}}_c^H} {\bm{h}}_{c,k} {\bm{h}}_{c,k}^H{\bm{w}}_c
	\end{array} & \\ 
	\mbox{s.t.}\quad
	&  {{\bm{w}}_c^H} {\bm{h}}_{r,j} {\bm{h}}_{r,j}^H{\bm{w}}_c \ge \Gamma^{th}, & \tag{\ref{P1.1}a}\\
	& \|{{\bm{w}}_c}\|^2 \le P_{\max}. & \tag{\ref{P1.1}b}
\end{alignat} 
Although problem (\ref{P1.1}) is a non-convex optimization problem, we show that it is able to derive the optimal beamforming vector in a closed-form expression and this also facilitates the subsequent UAV trajectory optimization.

\begin{table}[t]
	\caption{Important notations and symbols used in this work.} 
	\label{Notation}
	\footnotesize
	\centering
	\begin{tabular}{ll}%
		\hline
		$\rm{\textbf{Notation}}$  & $\rm{\textbf{ Physical meaning}}$ \\
		\hline
		${\bm{u}}_k$, ${\bm{v}}_j$ &  Location of user $k$ and target $j$   \\
		\hline
		${\bm{q}}[n]$             &  UAV's location at the $n$th time slot   \\
		\hline  
		$H$                       &  Altitude of the UAV\\
		\hline  
		$V_{\max}$                &  Maximum speed of the UAV \\
		\hline  
		$T_L$                     &  Time length of each ISAC frame  \\
		\hline  
		$N_L$                     &  Time slot number of each ISAC frame  \\
		\hline  
		$\delta_t$                &  Time interval of discrete locations \\
		\hline  
		$\Gamma^{th}_j$           &  Threshold of beam pattern gain for target $j$  \\
		\hline  
		$\alpha_{k}[n]$           &  Variable indicating whether user $k$ is served at time slot $n$\\
		\hline  
		$c_j[n]$                  &  Variable indicating whether target $j$ is sensed at time slot $n$ \\
		\hline  
		${\bm{w}}_c$  &  Beamforming vectors of communication signal  \\
		\hline  
		${\bm{x}}[n]$             &  Complex baseband transmitted signal  \\
		\hline  
		$R^{th}_k$                &  Minimum constraint of the achievable rate of user $k$  \\
		\hline  
		${\bm{A}}$                &  User association matrix \\
		\hline 
		${\bm{C}}$                &  Sensing time selection matrix \\
		\hline
		${\bm{Q}}$                &  UAV's trajectory vector \\
		\hline 
	\end{tabular}
\end{table}

\begin{Pro}\label{OptimalBearfoming}
	When $ c_j[n] = 1$ and $\alpha_{k}[n] = 1$, for any given UAV location ${\bm{q}}[n]$, the optimal beamforming vector can be expressed as 
	\begin{equation}\label{OptimalClosedBeamforming}
		{\bm{w}}_c^* \!=\! \left\{ {\begin{array}{*{20}{c}}
				{\sqrt{P_{\max}}\frac{{\bm{h}}_{c,k}}{\|{\bm{h}}_{c,k}\|},}&{ \tilde{\Gamma } \ge  {\Gamma ^{th}}}\\
				{\frac{1}{{{\lambda _1}}}( {\sqrt {{\beta_{c,k}}}{{\bm{h}}_{c,k}}  \!+\! {\lambda _2} \sqrt {{\Gamma^{th}}} {{\bm{h}}_{r,j}}  {e^{ - \jmath {\varphi _{k,j}}}}} ),}&{{\rm{Otherwise}}}
		\end{array}} \right.,
	\end{equation}
where ${ \varphi _{k,j}} = \arccos \frac{|{\bm{h}}_{c,k}^H {\bm{h}}_{r,j}|}{\| {\bm{h}}_{c,k}^H \| \| {\bm{h}}_{r,j} \|}$, ${\lambda _1} = \frac{\Upsilon{{{\| {{\bm{h}}_{c,k}^H} \|}^2}\sin { \varphi _{k,j}} }}{{\sqrt {{P_{\max }}{{\| {{{\bm{h}}_{r,j}}} \|}^2} - {\Gamma ^{th}}} }}$,  ${\lambda _2} = \frac{{  \Upsilon{{{\| {{\bm{h}}_{c,k}^H} \|}^2}\sqrt {{\Gamma ^{th}}}  - \Upsilon ^2 \| {{\bm{h}}_{c,k}^H} \|\| {{{\bm{h}}_{r,j}}} \|\cos { \varphi _{k,j}} }}}{{{{\| {{{\bm{h}}_{r,j}}} \|}^2}\sqrt {{P_{\max }}{{\| {{{\bm{h}}_{r,j}}} \|}^2}{\Gamma ^{th}} - {{( {{\Gamma ^{th}}} )}^2}}\sin { \varphi _{k,j}} }}$, ${\beta_{c,k}}{{ = }}\frac{{\| {{\bm{h}}_{c,k}^H} \|^2}}{{\| {{{\bm{h}}_{r,j}}} \|^2}}{\Upsilon ^2}$, $\Upsilon =  {\sqrt {{\Gamma^{th}}}}\cos  {\varphi _{k,j}} +  {\sqrt {{P_{\max }}{{\| {{{\bm{h}}_{r,j}}} \|}^2} - {\Gamma^{th}}}}\sin  {\varphi _{k,j}}$, and $\tilde{\Gamma } = \frac{{{M{P_{\max }}\cos ^2 { \varphi _{k,j}}}}}{d({\bm{q}}[n], {\bm{v}}_j)^2}$.
\end{Pro}
\begin{proof}
	Please refer to Appendix A.
\end{proof}

\par
In Proposition \ref{OptimalBearfoming}, the optimal beamforming vector could be intuitively viewed as two linearly superimposed beams towards user and target, respectively, which directly shows the influencing factors of the associated user's achievable rate. Also, the closed-form beamforming in (\ref{OptimalClosedBeamforming}) can also hold for arbitrary user channels ${{\bm{h}}_{c,k}^H}$. For $ \frac{{{M{P_{\max }}\cos ^2 { \varphi _{k,j}}}}}{d({\bm{q}}[n], {\bm{v}}_j)^2} <  {\Gamma ^{th}}$, the optimal SNR at user $k$ can be obtained by plugging ${\bm{h}}_{c,k}$ and ${\bm{h}}_{r,j}$ into $\beta_{c,k}$, yielding 
\begin{equation}\label{OptimalGamma}
	\begin{aligned}
		\gamma^*_{k,j} &= \gamma_0 \frac{{{d({\bm{q}}[n],{\bm{v}}_j)^2}}}{d({\bm{q}}[n],{\bm{u}}_k)^2} \\
		&{\left( {{  \sqrt {{\Gamma ^{th}_j}}}\cos { \varphi _{k,j}} \!+\! { \sqrt {\frac{{M{P_{\max }}}}{{d({\bm{q}}[n],{\bm{v}_j})^2}} \!-\! {\Gamma ^{th}}}\sin { \varphi _{k,j}}}   } \right)^2},
	\end{aligned}
\end{equation}
where $\gamma_0 = \frac{{\beta _0}}{\sigma ^2}$. 
\begin{remark}
	In (\ref{OptimalGamma}), the optimal user SNR is mainly determined by two parts: $\sqrt {{\Gamma ^{th}_j}}$ and $\sqrt {\frac{{M{P_{\max }}}}{{d({\bm{q}}[n],{\bm{v}_j})^2}}- {\Gamma ^{th}}}$, together with the channel correlation coefficient, i.e., $\cos { \varphi _{k,j}}$. When $\cos { \varphi _{k,j}} = 1$, the communication channel and target channel are linearly related. In this case, the channel power gain at user $k$ is $\frac{P_{\max}M \beta_0}{d({\bm{q}}[n],{\bm{u}}_k)^2}$, which holds if and only if the locations of user and target coincide. Whereas when $\cos { \varphi _{k,j}} = 0$, the communication channel and target channel are orthogonal to each other. In this case, the channel power gain at user $k$ is reduced to $\beta_0 \frac{{ {M{P_{\max }} - {\Gamma ^{th}} {d({\bm{q}}[n],{\bm{v}}_j)^2}} }}{d({\bm{q}}[n],{\bm{u}}_k)^2}$. 
\end{remark}

\begin{thm}\label{InftyAntenaNumber}
	If $M_x \to \infty$ and $M_y \to \infty$, for any given UAV location ${\bm{q}}[n]$, the optimal user $k$'s SNR during sensing target $j$ is denoted by
	\begin{equation}\label{InftyMBeta}
		\gamma^*_{k,j} = \left\{ {\begin{array}{*{20}{c}}
				{\gamma_0 \frac{{ {M{P_{\max }} - {\Gamma ^{th}}{d({\bm{q}}[n],{\bm{v}}_j)^2}} }}{{d({\bm{q}}[n],{\bm{u}}_k)^2}},}&{{\bm{u}}_k \ne {\bm{v}}_j}\\
				{\gamma_0 \frac{{ {M{P_{\max }} } }}{{d({\bm{q}}[n],{\bm{u}}_k)^2}},}&{{\rm{Otherwise}}}
		\end{array}} \right.,
	\end{equation}
	where $\gamma_0 = \frac{{\beta _0}}{\sigma ^2}$. And, the corresponding optimal UAV location with the maximum achievable rate at user $k$ during sensing target $j$ is given by 
	\begin{equation}
		{\bm{q}}^*_{k,j} = {\bm{u}}_k + \frac{{\sqrt {{{ Z}^2} + 4{H^2}} -Z}}{{2D_{k,j}}}({{\bm{v}}_j} - {{\bm{u}}_k}),
	\end{equation}
	where ${Z} = \frac{{M{P_{\max }} }}{{{\Gamma ^{th}}D_{k,j}}} - {D_{k,j}}$ and $D_{k,j} = \| {{\bm{v}}_j} - {{\bm{u}}_k} \|$ denotes the horizontal distance between user $k$ and target $j$.
\end{thm}
\begin{proof}
	Please refer to Appendix B.
\end{proof}

According to Lemma \ref{InftyAntenaNumber}, the user $k$'s SNR can be simplified as (\ref{InftyMBeta}) when the number of antennas is large, since the channel ${\bm{h}}^H_{c,k}$ and ${\bm{h}}^H_{r,j}$ can be completely irrelevant. However, solving (P1) is still very challenging due to the closely coupled integer variables and highly non-convex constraints. In the next subsection, we derive a tight lower bound of the achievable rate according to the optimal beamforming vector in Proposition \ref{OptimalBearfoming} to facilitate solving the problem (P1). 

\subsection{Lower Bound of Achievable Rate}
\label{LowerBoundRate}
For any given user association ${\bm{A}}$, sensing time selection $\bm{C}$, and UAV trajectory $\bm{Q}$, the optimal beamforming vector ${\bm{w}}_c$ can be obtained based on Proposition \ref{OptimalBearfoming}. Then, its corresponding achievable rate of user $k$ at the $n$th time slot is given by
\begin{equation}
	\begin{aligned}
		R_{k}[n] = &  \underbrace{\alpha_{k}[n]\left(1-\sum\nolimits_{j = 1}^Jc_j[n] \right) R^{C}_k[n]}_{\text{Only communication}} \\
		& + \underbrace{\alpha_{k}[n]\sum\nolimits_{j = 1}^J c_j[n]  R^{ISAC}_{k,j}[n]}_{\text{During sensing}},
	\end{aligned}
\end{equation}
where the user $k$'s optimal achievable rate during communication-only time is given by 
\begin{equation}\label{gamma_different_location}
	R^{C}_k[n] = \log_2\left( 1+\gamma_0 \frac{{M{P_{\max }}}}{{d({\bm{q}}[n],{\bm{u}}_k)^2}} \right),
\end{equation}
and the user $k$'s optimal achievable rate during sensing time is given by 
\begin{equation}\label{R_ISAC}
	R^{ISAC}_{k,j}[n] \!=\! \left\{ {\begin{array}{*{20}{c}}
			{\log_2\left( 1 \!+\!\gamma_0 \frac{{M{P_{\max }}}}{{d({\bm{q}}[n],{\bm{u}}_k)^2}} \right),}&\Gamma_{k,j}[n] \ge  {\Gamma ^{th}}\\
		{\log_2\left( 1 + \gamma_{k,j}^* \right),}&{\rm{Otherwise}}
\end{array}} \right.,
\end{equation}  
where $\Gamma_{k,j}[n]  = {\frac{{{M{P_{\max }}\cos^2 { \varphi _{k,j}}}}}{d({\bm{q}}[n], {\bm{v}}_j)^2}}$ and $\gamma_{k,j}^*$ is defined in (\ref{OptimalGamma}). Hence, the sum achievable rate can be maximized by only jointly optimizing the user association ${\bm{A}}$, sensing time selection $\bm{C}$, and UAV trajectory $\bm{Q}$. Nonetheless, the considered problem is still challenging due to the piece-wise non-concave function in (\ref{R_ISAC}). To handle this problem, a tight lower bound of $R^{ISAC}_{k,j}[n]$ is derived as below.

\begin{thm}\label{OptimalBeamformingLower}
	The optimal achievable rate of user $k$ during sensing target $j$ satisfies the following condition:
	\begin{equation}\label{GammaLowerBound}
		\begin{aligned}
			R^{ISAC}_k[n] &\ge  \log_2\left( 1+\gamma_0 \frac{{{{{{M{P_{\max }}}}  {- d({\bm{q}}[n],{\bm{v}}_j)^2}{\Gamma ^{th}}}}}}{d({\bm{q}}[n],{\bm{u}}_k)^2}  \right) \\
			&= \underline{R}^{ISAC}_{k,j}[n].
		\end{aligned}
	\end{equation}
\end{thm}
\begin{proof}
	To prove (\ref{GammaLowerBound}), we only need to ensure that $\gamma^*_{k,j} \ge \gamma_0 \frac{{{{{{M{P_{\max }}}}  {- d({\bm{q}}[n],{\bm{v}}_j)^2}{\Gamma ^{th}}}}}}{d({\bm{q}}[n],{\bm{u}}_k)^2}$ holds since the log function is a monotonically increasing function. If $\alpha_{k}[n] = 1$ and $\sum\nolimits_{j = 1}^Jc_j[n] = 0$, or $\frac{{{M{P_{\max }}\cos^2 { \varphi _{k,j}}}}}{d({\bm{q}}[n], {\bm{v}}_j)^2} \ge  {\Gamma ^{th}}$, the maximum ratio transmission (MRT) is the optimal beamforming vector to problem (P1), and thus, the inequality in (\ref{GammaLowerBound}) obviously holds. In the following, we prove that if $\frac{{{M{P_{\max }}\rho^2}}}{d({\bm{q}}[n], {\bm{v}}_j)^2} <  {\Gamma ^{th}}$, ${\rho  {\sqrt {{\Gamma ^{th}}}} {{ + }}\sqrt {\left( {1 - {\rho ^2}} \right)} {\sqrt {\frac{{M{P_{\max }}}}{{d({\bm{q}}[n],{\bm{v}}_j)^2}}- {\Gamma ^{th}}}}   } \ge {\sqrt {\frac{{M{P_{\max }}}}{{d({\bm{q}}[n],{\bm{v}}_j)^2}}- {\Gamma ^{th}}}}$. Let $\rho = \cos { \varphi _{k,j}}$ and $G = \frac{{M{P_{\max }}}}{{d({\bm{q}}[n],{\bm{v}}_j)^2}}$ for notation simplicity. Then, for ${\cal{F}}(\Gamma ^{th},\rho) \triangleq \rho  {\sqrt {{\Gamma ^{th}}}}  + \sqrt {\left( {1 - {\rho ^2}} \right)} {\sqrt {G- {\Gamma ^{th}}}}   - \sqrt{G- \Gamma ^{th}}$, we need to prove ${\cal{F}}(\Gamma ^{th},\rho) \ge 0$ for $\Gamma ^{th} \in ({G}{\rho}, G]$. As ${\cal{F}}(\Gamma ^{th},\rho)$ is an increasing function with respect to (w.r.t) $\Gamma ^{th}$, ${\cal{F}}(\Gamma ^{th},\rho) \ge 0$ if ${\cal{F}}(G \rho,\rho) \ge 0$, where ${\cal{F}}(G \rho,\rho) = {\sqrt {G}}\left(\rho   {{ + }}\sqrt {\left( {1 - {\rho ^2}} \right)} {\sqrt {1-\rho}   } - \sqrt{1-\rho}\right)$. For $\rho \in [0,1]$, ${\cal{F}}(G \rho,\rho)$ is an increasing function w.r.t $\rho$, as $\frac{\partial{\cal{F}}(G \rho,\rho)}{\partial \rho} \!>\! 0$. Hence, ${\cal{F}}(G \rho,\rho) \ge {\cal{F}}(0,0) \!=\! 0$. Then, plugging ${\cal{F}}(G \rho,\rho) \!\ge\! 0$ into (\ref{OptimalGamma}), we obtain $\gamma^* \!\ge\! \gamma_0 \frac{{{{{{M{P_{\max }}}}  {- d({\bm{q}}[n],{\bm{v}}_j)^2}{\Gamma ^{th}}}}}}{d({\bm{q}}[n],{\bm{u}}_k)^2}$, which thus completes the proof.
\end{proof}

The lower bound of user $k$'s SNR in (\ref{GammaLowerBound}) is tight if $M$ goes to infinity according to Lemma \ref{InftyAntenaNumber}. A closer look at this lower bound in (\ref{GammaLowerBound}) reveals that the value of $\cos { \varphi _{k,j}}$ is small since $ {\frac{{\sin M  \Delta \pi /2}}{{ \sin  \Delta \pi /2}}} $ is relatively small for $\Delta \ge \frac{1}{M} $, and $R^{ISAC}_k[n]  = \underline{R}^{ISAC}_{k,j}[n]$ when $\Delta = \frac{2i}{M}$, $i \in {\mathbb{Z}}, i \ne 0$. Based on Lemma \ref{OptimalBeamformingLower}, the lower bound of the user $k$'s achievable rate can be recast as
\begin{equation}\label{LowerBoundAchievable}
	\begin{aligned}
		\underline R_{k}[n] =& \alpha_{k}[n] R^{C}_k[n] \\
		&+  \sum\nolimits_{j = 1}^J \alpha_{k}[n]c_j[n] \left(\underline R^{ISAC}_{k,j}[n] \!-\! R^{C}_k[n]\right).
	\end{aligned}
\end{equation}
Then, we introduce problem (P1.1) as the lower bound of the achievable rate maximization problem in the case by setting $R_{k}[n]$ as $\underline R_{k}[n]$ in (P1). Then, a high-quality solution of problem (P1) can be obtained by solving problem (P1.1), elaborated as follows.

\subsection{Penalty-based Problem Transformation}
\label{PenaltyAlgorithm}
\par 
Although the complicated expression of the optimal achievable rate of user $k$ is simplified as its tight lower bound, the integer variables $\{\alpha_{k}[n]\}$ and $\{c_j[n]\}$ are coupled with each other in the objective function and constraints. To tackle this issue, another variable $e_{k,j}[n] = \alpha_{k}[n]  c_j[n]$ is introduced to decouple the integer variables. Then, $\underline R_{k}[n]$ can be rewritten as 
\begin{equation}
	\underline R_{k}[n] = \alpha_{k}[n] R^{C}_k[n] + 
	\sum\nolimits_{j = 1}^J e_{k,j}[n] \left(\underline R^{ISAC}_{k,j}[n] - R^{C}_k[n]\right),
\end{equation}
where $e_{k,j}[n] \in \{0,1\}$. To ensure the consistency of the problem (P1.1), some other constraints are introduced to replace that in (\ref{SensingFrequencyConstraintsA}) and (\ref{SensingFrequencyConstraintsB}) as follows
\begin{equation}\label{ConstraintUserAssociation}
	\alpha_{k}[n] \ge  e_{k,j}[n], \forall k, j, n,
\end{equation}
\begin{equation}\label{ConstraintTargetSensing}
	\sum\nolimits_{n = (l - 1){N_L} + 1}^{l  {N_L}} \sum\nolimits_{k = 1}^K	e_{k,j}[n] = 1, \forall l, j,
\end{equation}
\begin{equation}\label{ConstraintSesingSlot}
	\sum\nolimits_{k = 1}^K	\sum\nolimits_{j = 1}^J e_{k,j}[n] \le 1, \forall n.
\end{equation}
(\ref{ConstraintUserAssociation}) ensures that $e_{k,j}[n] = 1$ if and only if $\alpha_{k}[n] = 1$. Accordingly, we can readily prove that the new introduced problem with the replaced constraints (\ref{ConstraintUserAssociation})-(\ref{ConstraintSesingSlot}), denoted by (P1.2), is equivalent to (P1.1). Furthermore, the bream pattern gain constraints in (\ref{P1}a) can be transformed into
\begin{equation}\label{NewBeamPatternConstraints}
	\sum\nolimits_{k = 1}^{K}{{e_{k,j}}[n]}  (M{P_{\max }} - d({\bm{q}}[n], {\bm{v}}_j)^2  {\Gamma ^{th}_j}) \ge 0 .
\end{equation}
However, converting $\alpha_k[n]$ and $e_{k,j}[n]$ to continuous-valued variables and then utilizing rounding function to obtain the binary solution, generally may not  satisfy the QoS constraints in (\ref{P1}c) and the beam pattern gain constraints in (\ref{P1}a). Several slack matrices ${\bar {\bm{A}}} = \{ \{\bar \alpha_k[n]\}_{n=1}^N \}_{k=1}^K$ and ${\bar {\bm{E}}} = \{ \{\{\bar e_{k,j}[n]\}_{n=1}^N\}_{k=1}^K \}_{j=1}^J$ are presented to transform the binary constraints into a series of equivalent equality constraints. Specifically, (\ref{P1}b) can be rewritten as
\begin{equation}\label{PenltyA}
	\alpha_k[n](1 - \bar \alpha_k[n]) = 0, \quad \alpha_k[n] = \bar \alpha_k[n], \quad \forall  k,  n, 
\end{equation}
\begin{equation}\label{PenltyB}
	e_{k,j}[n](1 - \bar e_{k,j}[n]) = 0, \quad e_{k,j}[n] = \bar e_{k,j}[n], \quad  \forall k, j, n. 
\end{equation}
We can readily derive that $\alpha_k[n]$ and $e_{k,j}[n]$ satisfying the above two constraints must be either 1 or 0, which confirms the equivalence of the transformation of (\ref{P1}b) into these two constraints. Then, (\ref{PenltyA}) and (\ref{PenltyB}) are added to the objective function in (P1.2) as the penalty terms \cite{bertsekas1997nonlinear}, yielding the following optimization problem 
\begin{alignat}{2}
	\label{P2}
	(\rm{P2}): \quad & \begin{array}{*{20}{c}}
		\mathop {\min }\limits_{\bar{\bm{A}}, \bar{\bm{E}}, {\bm{A}}, {\bm{E}}, {{\bm{Q}}}} \quad - \underline{R}
	\end{array} & \\ 
	\mbox{s.t.}\quad
	& (\ref{UserAssociation}), (\ref{ConstraintUserAssociation})-(\ref{NewBeamPatternConstraints}), (\ref{P1}c)-(\ref{P1}d), \nonumber \\
	& \frac{1}{N_L}\sum\nolimits_{n = (l - 1){N_L} + 1}^{l  {N_L}}  \alpha_{k}[n] \underline R_{k}[n] \ge {R}_k^{th}, \forall k, l,  & \tag{\ref{P2}a} 
\end{alignat} 
where $\underline{R}$ is defined in (\ref{Rdefine})
\begin{figure*}[b]
	\begin{equation}\label{Rdefine}
		\begin{aligned}
			\underline{R}   =& \frac{1}{N} \sum\nolimits_{n = 1}^{N} \sum\nolimits_{k = 1}^K    \alpha_{k}[n] \underline R_{k}[n]  -  \frac{1}{2\eta} \sum\nolimits_{n = 1}^{N}  \sum\nolimits_{k = 1}^K ( | \alpha_k[n](1 - \bar \alpha_k[n]) |^2 + | \alpha_k[n] - \bar \alpha_k[n]|^2  ) \\
			&- \frac{1}{2\eta} \sum\nolimits_{n = 1}^{N} \sum\nolimits_{j = 1}^J \sum\nolimits_{k = 1}^K   ( | e_{k,j}[n](1 - \bar e_{k,j}[n]) |^2 + | e_{k,j}[n] - \bar e_{k,j}[n]|^2 ),
		\end{aligned}
	\end{equation}
\end{figure*}
and $\eta > 0$ is the penalty coefficient used to penalize the violation of the equality constraints (\ref{PenltyA}) and (\ref{PenltyB}). Despite relaxing the equality
constraints in (\ref{PenltyA}) and (\ref{PenltyB}), it can be readily verified that the solutions obtained will always satisfy the equality constraints (i.e., binary value constraints of $\{\alpha_{k}[n]\}$ and $\{e_{k,j}[n]\}$), when $\frac{1}{\eta} \to \infty$. To facilitate efficient optimization, $\eta$ is initialized with a sufficiently large value and then we gradually reduce $\eta$ to a sufficiently small value. As a result, a feasible binary solution can be eventually obtained. In particular, the alternating optimization (AO) method is applied to iteratively optimize the primary variables in different blocks, as shown in Section \ref{InnerLayerAlgorithm}. 

\subsection{Inner and Outer layer Iteration}
\label{InnerLayerAlgorithm}
In this subsection, we propose a two-layer penalty-based algorithm. Specifically, in the inner layer, (P2) is divided into three sub-problems in which $\{\bar {\bm{A}}, \bar {\bm{E}}\}$, $\{ {\bm{A}},  {\bm{E}}\}$, and $\bm{Q}$ are optimized iteratively. In the outer layer, the penalty coefficient is updated to ensure that the constraints (\ref{PenltyA}) and (\ref{PenltyB}) are met eventually.

\subsubsection{Slack Variables Optimization} 
For any given $\{ {\bm{A}},  {\bm{E}}\}$ and $\bm{Q}$, (P2) can be expressed as
\begin{alignat}{2}
	\label{P2.1.1}
	(\rm{P2.1}): \quad & \begin{array}{*{20}{c}}
		\mathop {\min }\limits_{\bar{\bm{A}}, \bar{\bm{E}}} \quad   - \underline{R}
	\end{array} &  
\end{alignat} 
It is not difficult to find that the slack variables $\bar \alpha_k[n]$ and $\bar e_{k,j}[n]$ are only involved in the objective function. Thus, the optimal slack variables $\bar \alpha_k[n]$ and $\bar e_{k,j}[n]$ can be obtained by setting the derivative of (\ref{P2.1.1}) w.r.t. $\bar \alpha_k[n]$ and $\bar e_{k,j}[n]$ to zero, respectively, i.e.,
\begin{equation}\label{SlakeVariableA}
	\bar{\alpha}_{k}^{\mathrm{opt}}[n]=\frac{\alpha_{k}[n]+\alpha_{k}^{2}[n]}{1+\alpha_{k}^{2}[n]}, \forall k, n,
\end{equation}
\begin{equation}\label{SlakeVariableB}
	\bar{e}_j^{\mathrm{opt}}[n]=\frac{e_{k,j}[n]+e_{k,j}^{2}[n]}{1+e_{k,j}^{2}[n]}, \forall j, k, n.
\end{equation}
\subsubsection{User Association and Sensing Time Selection}
For any given $\{ \bar{\bm{A}},  \bar{\bm{E}}\}$ and $\bm{Q}$, (P2) can be expressed as
\begin{alignat}{2}
	\label{P2.1.2}
	(\rm{P2.2}): \quad & \begin{array}{*{20}{c}}
		\mathop {\min }\limits_{{\bm{A}}, {\bm{E}}} \quad   - \underline{R}
	\end{array} & \\ 
	\mbox{s.t.}\quad
	& 	 (\ref{UserAssociation}), (\ref{ConstraintUserAssociation})-(\ref{NewBeamPatternConstraints}), ({\ref{P2}a}). \nonumber
\end{alignat} 
It can be seen that problem (\ref{P2.1.2}) is convex with a quadratic objective function and linear inequality constraints, which can be solved by standard convex optimization solvers, such as CVX.

\subsubsection{Trajectory Optimization}
For given $\{\bar {\bm{A}}, \bar {\bm{E}}\}$ and $\{ {\bm{A}},  {\bm{E}}\}$, the UAV trajectory optimization sub-problem is given as follows
\begin{alignat}{2}
	\label{P2.1.4}
	(\rm{P2.3}): \quad & \begin{array}{*{20}{c}}
		\mathop {\max }\limits_{ {{\bm{Q}}}} \quad \mathop \frac{1}{N} \sum\nolimits_{n = 1}^{N} \sum\nolimits_{k = 1}^K  \underline R_{k}[n]
	\end{array} & \\ 
	\mbox{s.t.}\quad
	&  ({\ref{P1}}e), ({\ref{P1}}f), (\ref{NewBeamPatternConstraints}), (\ref{P2}a).  & \nonumber
\end{alignat} 

However, note that (P2.3) is neither concave or quasi-concave due to the non-convex constraints (\ref{P2}a), (\ref{P1}a) and the non-convex objective function (\ref{P2.1.4}). In general, there is no efficient method to obtain the optimal solution. In the following, we adopt the successive convex optimization technique to solve (P2.3). To this end, additional slack variables $\{z_{c,k}[n]\}$ and $\{z_{r,j}[n]\}$ are introduced, and $R^{C}_k[n]$ and $\underline {R}^{ISAC}_{k,j}[n]$ are recast as
\begin{equation}
	\tilde{R}^{C}_k[n] = B \log_2 \left(1 + \beta_{0}\frac{{P}_{\max} M    }{z_{c,k}[n]} \right),
\end{equation}
\begin{equation}
	\tilde{R}^{ISAC}_{k,j}[n] \!=\! \log_2\left( 1 \!+\! \gamma_0 \frac{{{M{P_{\max }} - z_{r,j}[n]{\Gamma ^{th}}}}}{{z_{c,k}[n] }}  \right),
\end{equation}
together with 
\begin{equation}\label{SlackZA}
	z_{c,k}[n] \ge {\|{{\bm{q}}[n] - {\bm{u}}_k}\|^2} + {H^2}, \forall k, n,
\end{equation}
\begin{equation}\label{SlackZB}
	z_{r,j}[n] \ge {\|{{\bm{q}}[n] - {\bm{v}}_j}\|^2} + {H^2}, \forall k, j,  n.
\end{equation}
For ease of analysis, this new constructed problem is denoted by (P2.4). It can be shown that at the optimal solution of variable $\tilde{R}^{C}_k[n]$ and $\tilde {R}^{ISAC}_{k,j}[n]$ in (P2.4), all the constraints in (\ref{SlackZA}) and (\ref{SlackZB}) are active, since otherwise we can always increase $z_{c,k}[n]$ or $z_{c,k}[n]$ without decreasing the value of the objective function. Hence, (P2.4) is equivalent to (P2.3). Since $\tilde R^{C}_k[n]$ is convex w.r.t. ${z_{c,k}[n] }$, for any local point ${z^{(r)}_{c,k}[n] }$ obtained at the $r$th iteration, we have 
\begin{equation}\label{RCRateExpression}
\begin{aligned}
\tilde{R}^{C} _k[n] =& \log_2 \left( 1 + \frac{A_k}{z_{c,k}[n]} \right) \ge \log_2 \left( 1 + \frac{A_k}{z^{(r)}_{c,k}[n]} \right) \\
& - \frac{A_k\left(z_{c,k}[n] - z^{(r)}_{c,k}[n] \right)}{(z^{(r)}_{c,k}[n]^2 + A_k z^{(r)}_{c,k}[n]) \ln 2}  = \hat{R}^{C} _k[n],	
\end{aligned}
\end{equation}
where $A_k = \frac{P_{\max} {M } {{\beta_0}}}{   {{\sigma _k^2} } }$. Then, $\tilde{R}^{ISAC}_{k,j}[n] = \log_2\left( {z_{c,k}[n] } \!+\! \gamma_0 {{{M{P_{\max }} - \gamma_0 z_{r,j}[n]{\Gamma ^{th}}}}}{}  \right) -  \log_2\left( {z_{c,k}[n] } \right)$. By introducing a new variable ${u_{c,k}[n] }$ stratifying   $ {u_{c,k}[n] } \le {z_{c,k}[n] } \!+\! \gamma_0 {{{M{P_{\max }} - z_{r,j}[n]{\Gamma ^{th}}}}}$, we have $\tilde{R}^{ISAC}_{k,j}[n] \ge \log_2\left({u_{c,k}[n] } \right) -  \log_2\left( {z_{c,k}[n] } \right)$. Similarly,  $\log_2\left( {z_{c,k}[n] } \right)$ can be transformed into a linear function of $z_{c,k}[n]$ using the same method applied in (\ref{RCRateExpression}). Then, the transformed $\tilde R^{ISAC}_{k,j}[n]$ is denoted by $\hat R^{ISAC}_{k,j}[n] = \log_2\left({u_{c,k}[n] } \right) - \log{ z^{(r)}_{c,k}[n]} - \frac{1}{ z^{(r)}_{c,k}[n] \ln 2 }\left(z_{c,k}[n] - z^{(r)}_{c,k}[n]\right)$, and (P2.4) can be converted into 
\begin{alignat}{2}
	\label{P2.1.5}
	(\rm{P2.5}): \quad & \begin{array}{*{20}{c}}
		\mathop {\max }\limits_{ {{\bm{Q}}}, \{{z}_{c,k}\}, \{z_{r,j}\}, \{u_{c,k}\}} \quad \mathop \frac{1}{N} \sum\nolimits_{n = 1}^{N} \sum\nolimits_{k = 1}^K    \hat R_k[n]
	\end{array} & \\ 
	\mbox{s.t.}\quad
	& ({\ref{P1}}e), ({\ref{P1}}f), (\ref{NewBeamPatternConstraints}),  (\ref{SlackZA}), (\ref{SlackZB}),  & \nonumber \\
	& \frac{1}{N_L}\sum\nolimits_{n = (l - 1){N_L} + 1}^{l  {N_L}}  \hat R_k[n] \ge {R}_k^{th}, \forall k, l, & \tag{\ref{P2.1.5}a} 
\end{alignat} 
where $\hat R_k[n] =  \alpha_{k}[n]  \hat R^{C}_k[n] +
\sum\nolimits_{j = 1}^J e_{k,j}[n] (\hat R^{ISAC}_{k,j}[n] - \tilde R^{C}_k[n])$. Based on the previous discussions, all of the constraints of (P2.5) are convex constraints. Thus, (P2.5) is a convex optimization problem that can be efficiently solved by convex optimization solvers such as CVX.

\subsubsection{Outer layer Iteration}

In the outer layer, the value of the penalty coefficient $\eta$ is gradually decreased by updating $\eta = z \eta$, where $z$ ($0 < z < 1$) is a scaling factor. A larger value of $z$ can achieve better performance but at the cost of more iterations in the outer layer.

\subsection{Convergence Analysis and Computational Complexity}%
To show the converged solutions of the proposed penalty-based algorithm, the terminal criteria for the outer layer is given as ${\max} ( |\alpha_k[n](1 - \bar \alpha_k[n])|, |  \alpha_k[n] - \bar \alpha_k[n]|, |e_{k,j}[n](1 - \bar e_{k,j}[n])|, |  e_{k,j}[n] - \bar e_{k,j}[n]|, \forall k, j, n ) \le \xi$, where $\xi $ is a predefined accuracy. The details of the proposed penalty-based algorithm are shown in {\bf{Algorithm} \ref{PenaltyBasedAlgorithm}}. In the inner layer, with the given penalty coefficient, the objective function of (P2) is non-increasing over each iteration during applying the AO method and the objective of (P2) is upper bounded due to the limited flying time $T$ and transmit power $P_{\max}$. As such, a stationary point can be achieved in the inner layer. In the outer layer, the penalty coefficient is gradually decreased so that the equality constraints (\ref{PenltyA}) and (\ref{PenltyB}) are ultimately satisfied. Based on Appendix B in \cite{Cai2017Joint}, this penalty-based framework is guaranteed to converge.

\begin{algorithm}[t]
	\small
	\caption{Penalty-Based Algorithm}
	\label{PenaltyBasedAlgorithm}
	\begin{algorithmic}[1]
		\STATE {\bf{Initialize}}  $\{{\bar {\bm{A}}^{(0)}}, {\bar {\bm{E}}^{(0)}}\}$, $\{{ {\bm{A}}^{(0)}}, { {\bm{E}}^{(0)}}\}$, and $\bm{Q}^{(0)}$, the iteration number $r = 1$, the convergence accuracy $\epsilon_1$ and $\epsilon_2$.
		\REPEAT 
		\REPEAT 
		\STATE With given $\{\{{ {\bm{A}}^{(r)}}, { {\bm{E}}^{(r)}}\}, \bm{Q}^{(r)}\}$, obtain $\{{\bar {\bm{A}}^{(r+1)}}, {\bar {\bm{E}}^{(r+1)}}\}$ based on (\ref{SlakeVariableA}) and (\ref{SlakeVariableB}).
		\STATE With given $\{\{{\bar {\bm{A}}^{(r)}}, {\bar {\bm{E}}^{(r)}}\}, \bm{Q}^{(r)}\}$, obtain $\{{ {\bm{A}}^{(r+1)}}, { {\bm{E}}^{(r+1)}}\}$ by solving the problem in (\ref{P2.1.2}).
		\STATE With given $\{\{{\bar {\bm{A}}^{(r)}}, {\bar {\bm{E}}^{(r)}}\}, \{{ {\bm{A}}^{(r)}}, { {\bm{E}}^{(r)}}\}\}$, and obtain $ \bm{Q}^{(r+1)}$ by solving the problem in (\ref{P2.1.5}). 
		\STATE Calculate $ C^{(r+1)*}$ according to the objective function of (P2).
		\STATE $r = r + 1$.
		\UNTIL $\left|C^{(r+1)*} - C^{(r)*}\right| \le \epsilon_1$
		\STATE $\eta = z \eta$.
		\UNTIL	the constraint violation in (\ref{PenltyA}) and (\ref{PenltyB}) is below a threshold $\epsilon_2$.
		\STATE Obtain ${\bm{w}}_c^*$ based on proposition {\ref{OptimalBearfoming}}.
		\STATE Recover optimal sensing time selection ${\bm{C}}^*$ based on ${ {\bm{A}}^*}$ and ${ {\bm{E}}^*}$.
	\end{algorithmic}
\end{algorithm}
\par 
The complexity of {\bf{Algorithm} \ref{PenaltyBasedAlgorithm}} can be analyzed as follows. In the inner layer, the main complexity of {\bf{Algorithm} \ref{PenaltyBasedAlgorithm}} comes from steps 5 and 6. In step 5, the complexity of computing $\{\alpha_{k}[n]\}$ and $\{e_{k,j}[n]\}$ is ${\cal{O}}(KN + JKN)^{3.5}$ \cite{zhang2019securing}, where $KN + JKN$ stands for the number of variables \cite{zhang2019securing}. Similarly, in step 6, the complexity required to compute the UAV trajectory is ${\cal{O}}(2N + KN + JN)^{3.5}$ \cite{zhang2019securing}, where $2N + KN + JN$ denotes the number of variables. Therefore, the total complexity of {\bf{Algorithm} \ref{PenaltyBasedAlgorithm}} is ${\cal{O}}( L_{outer}L_{inner}((KN + JKN)^{3.5} + (2N + KN + JN)^{3.5}))$, where $L_{inner}$ and $L_{outer}$ denote the number of iterations required for reaching convergence in the inner and outer layers, respectively.

\section{Analysis Without Location Constraints and Low-complexity Algorithm for Solving (P1)}
\label{WithoutLocationConstraint}
To draw important insights into periodic sensing and communication design, we further study a special case of (P1) where the initial and final location constraints are ignored, denoted by (P3). Specifically, (P3) is given as 
\begin{alignat}{2}
	\label{P3}
	(\rm{P3}): \quad & \begin{array}{*{20}{c}}
		\mathop {\max }\limits_{{{\bm{w}}_c}, {\bm{A}}, {{\bm{Q}}}, {\bm{C}}} \quad \mathop  \frac{1}{N}  \sum\nolimits_{n = 1}^N \sum\nolimits_{k = 1}^K \alpha_k[n] R_k[n]
	\end{array} & \\ 
	\mbox{s.t.}\quad
	&  (\ref{P1}a) - (\ref{P1}e). \nonumber
\end{alignat} 
In the following, we first present the structural characteristics of the optimal solutions in different ISAC frames of (P3). Based on this, a low-complexity algorithm to problem (P1) is proposed to solve (P1). 

\subsection{Analysis of Optimal Solution to (P3)}

For ease of analysis, denote ${\bm{{\cal{X}}}}_l[n] = \{{\bm{w}}_{c,l}^*[n], {\bm{\alpha}}_{l}[n]^*, {\bm{c}}_{l}[n]^*, {\bm{q}}^*_l[n]\}$ as the optimal solution of the $n$th time slot of the $l$th ISAC frame, where ${\bm{w}}_{c,l}^*[n]$, ${\bm{\alpha}}^*_l[n]$, ${\bm{c}}^*_{l}[n]$, and ${\bm{q}}^*_l[n]$ represent its corresponding optimal beamforming vector, user association, sensing time selection, and UAV trajectory at the $n$th time slot.

\begin{thm}\label{EqualForEachFrame}
	There always exists an optimal solution to problem (P3) satisfying the following condition:
	\begin{equation}\label{OptimalConditionISACFrame}
		{\bm{{\cal{X}}}}_{l'}[n] =\left\{\begin{array}{ll}
			{\bm{{\cal{X}}}}_{l}[n] , & \left|l-l'\right|\mid 2 \\
			{\bm{{\cal{X}}}}_{l}[N_L - n + 1] , & \left|l-l'\right|\nmid 2
		\end{array}\right.,
	\end{equation}
	where the symbols $\mid$ and $\nmid$ represent that $\left|l-l'\right|$ is divisible and not divisible by 2, respectively, $n \in \{1,\cdots,N_L\}$, and $l$, $l' \in {\cal{L}}$.
\end{thm}

\begin{proof}
	Assume that at the optimal solution to problem (P3), the maximum sum achievable rate of the $l$th ISAC frame is denoted by $C^*_l$, its corresponding optimal beamforming vector, user association, sensing time selection, and UAV trajectory are denoted by $\{{\bm{{\cal{X}}}}_l[n]\}_{n=l}^{N_L}$. Without loss of generality, we assume that the sum achievable rate $C^*_l$ of the $l$th ISAC frame is the largest in the set $\{C^*_1,\cdots,C^*_L\}$. We can always obtain a solution of the $l'$th ISAC frame by reorganizing the elements in $\{{\bm{{\cal{X}}}}_l[n]\}_{n=l}^{N_L}$ while satisfying the constraints in (\ref{P1}a)-({\ref{P1}f}), and its corresponding sum achievable rate $ C^*_l \ge C^*_{l'}$. Specifically, considering the maximum speed constraint, when $l' = l + 2 i + 1$, $i \in \mathbb{Z}$, a solution whose achievable rate is no less than $C^*_l$ can be constructed by reversing the sequence of that within the $l$th ISAC frame, i.e., ${\bm{{\cal{X}}}}_{l'}[n] = {\bm{{\cal{X}}}}_{l}[N_L - n + 1]$, $n \in \{1,\cdots,N_L\}$. Similarly, when $l' = l + 2 i$, $i \in \mathbb{Z}$, we can readily prove that the solution of $l$th is also feasible for the $l'$th ISAC frame, i.e., ${\bm{{\cal{X}}}}_{l'}[n] = {\bm{{\cal{X}}}}_{l}[n]$. By combing the above results, there always exists an optimal solution to problem (P3) satisfying the condition in (\ref{OptimalConditionISACFrame}). This thus completes the proof.
\end{proof} 

\begin{remark}\label{RemarkSameFrame}
	According to Lemma {\ref{EqualForEachFrame}}, there always exists an optimal solution to problem (P3) in the $l$th ISAC frame, which is exactly equal or opposite in sequence to that of the $l'$th ISAC frame. Specifically, for any two time slot $n_1$ and $n_2$ belong to two adjacent ISAC frames, the optimal solution at time slot $n_1$ and that at time slot $n_2$ are equal when $n_1 + n_2 = l N_L +1$, i.e., $n_1$ and $n_2$ are symmetrical with respect to the time instant $l T_L / 2$, where $l$ is an even number. Hence, Lemma 4 implies that although the UAV trajectories within different ISAC frames are coupled with each other due to the maximum speed constraint, problem (P3) can be solved by only obtaining the solution in the first ISAC frame, while the solutions of other ISAC frames can be obtained based on (\ref{OptimalConditionISACFrame}). In particular, the solution of the first ISAC frame for problem (P3) can be efficiently solved by {\bf{Algorithm} \ref{PenaltyBasedAlgorithm}} due to the similar constraints and objective function. 
\end{remark}

\begin{Pro}\label{IncreaseNL}
	The maximum achievable rate in (P3) increases monotonically as $T_L$ increases. 
\end{Pro}
\begin{proof}
	Based on Lemma \ref{EqualForEachFrame}, there always exists an optimal solution, whose achievable rate in each ISAC frame is equal, denoted by $C^*_l$. For any given $T_L$, assume that at the optimal solution to problem (P3), the optimal beamforming, user association, sensing time slots, and UAV trajectory of the $l$th ISAC frame are denoted by ${\bm{w}}_{c,l}^*$, ${\bm{A}}^*_l$, ${\bm{C}}^*_{l}$, and ${\bm{Q}}^*_l$, respectively. Without loss of generality, the maximum achievable rate in the $l$th frame is denoted by $R^{\max}$, its corresponding time slot and UAV location are denoted by $n^{\max}$ and ${\bm{q}}^*[n^{\max}]$. Based on the above discussion, for $N'_L > N_L$, there always exists a solution, in which the UAV trajectory in the $l$th ISAC frame can be given by
	\begin{equation}\label{TrajectoryN_DeltaL}
		\begin{aligned}
			{\bm{Q}}'_l & \!=\! \{{\bm{q}}_l^*[1], \cdots, {\bm{q}}_l^*[n^{\max}-1], \underbrace{{\bm{q}}_l^*[n^{\max}],\cdots,{\bm{q}}_l^*[n^{\max}]}_{N'_L - N_L + 1},  \\
			& {\bm{q}}_l^*[n^{\max}+1], \cdots, {\bm{q}}_l^*[N_L]\},
		\end{aligned}
	\end{equation}
	and its corresponding beamforming, user association, and sensing time selection is set as the same with that of solution $\{\{{\bm{w}}_{c,l}^*, {\bm{w}}_{r,l}^*\}, {\bm{A}}_{l}^*, {{\bm{C}}}^*_{l}\}$ based on the UAV location. Let $\Delta l = \frac{N}{N_L} - \frac{N}{N'_L} \in \mathbb{Z}$. Then, the achievable rate based on the UAV trajectory in (\ref{TrajectoryN_DeltaL}) can be given by
	\begin{equation}
		\begin{aligned}
			&(L - \Delta l) \left( C^*_l + (N'_L - N_L) R^{\max} \right) \\
			=& (L - \Delta l)   C^*_l + \Delta l  N_L R^{\max}  \ge L  C^*_l.
		\end{aligned}
	\end{equation}
	Hence, the achievable rate with frame length $N'_L$ is no less than that with frame length $N_L$, thus completing the proof.
\end{proof}

In Proposition \ref{IncreaseNL}, we reveal a useful and fundamental trade-off between sensing frequency and communication rate. Note that the above interesting results not only help solve problem (P3) more efficiently but also provide a novel idea to construct a high-quality solution to problem (P1), as elaborated below. 

\subsection{Low-Complexity Algorithm for solving (P1)}
\label{LowComplexitySolution}
The large mission period $T$ may entail a large number of trajectory points in practice, thus resulting in prohibitive computational complexity for the UAV trajectory design. To handle this problem, a low-complexity method to (P1) is presented based on our derived structural characteristics among ISAC frames (c.f. Lemma \ref{EqualForEachFrame}). To facilitate the analysis, we introduce problem (P3.1) as the achievable rate maximization problem in the case without the initial and final location constraints, which can be expressed similarly as (P2) by removing constraint (\ref{P1}f). 

If the optimal achievable rate of problem (P3.1) is denoted by $R^*$, it is not difficult to find that the optimal achievable rate of problem (P1) equals to $R^*$ when $T \to \infty$. The optimized UAV trajectory of problem (P3.1) obtained via {\bf{Algorithm} \ref{PenaltyBasedAlgorithm}} is denoted by ${\bm{Q}}' = \{{\bm{q}}'[1], \cdots, {\bm{q}}'[N_L]\}$. Then, a high-quality and low-complexity UAV trajectory of (P1) can be obtained by composing three sub-trajectories: The UAV first flies straightly at its maximum speed from the initial location ${\bm{q}}_I$ towards ${\bm{q}}^*[1]$ or ${\bm{q}}^*[N_L]$ (closer one); then flies back and forth along the trajectory ${\bm{Q}}'$; finally flies straightly at its maximum speed to final location ${\bm{q}}_F$. Furthermore, the corresponding optimized sensing time selection and user association for this constructed UAV trajectory can be solved by {\bf{Algorithm} \ref{PenaltyBasedAlgorithm}} in a similar way. The complexity of this constructed solution is mainly determined by the step of obtaining the solution ${\bm{Q}}'$, which is about $\frac{L^{3.5}-1}{L^{3.5}} 100 \%$ percent reduced as compared to that of solving (P1) via {\bf{Algorithm} \ref{PenaltyBasedAlgorithm}} directly. In particular, this low-complexity algorithm is preferred when the number of frames $L$ is relatively large. 

\section{Numerical Results}
\label{Simulations}

\par 
In this section, numerical results are provided for characterizing the performance of the proposed periodic sensing and communication design and for gaining insights into the design and implementation of UAV-based ISAC systems. In the simulation, we consider an area of 1 km $\times$ 1 km with $K=4$ users and $J=4$ targets in the interested sensing area. Unless otherwise stated, the system parameters are set as follow. The number of antennas at the UAV $M = 16$ ($M_x = M_y = 4$), and the beam pattern gain threshold $\Gamma^{th} = 6e^{-5}$. The UAV's maximum horizontal flight speed is set as $V_{\max} = 30$ m/s with the flight altitude $H = 40$ m. In addition, the channel power gain at the reference distance $d_0 = 1$ m and the noise power at each user are set as $\beta_{0} = -30$ dB and $\sigma^2 = -100$ dB, respectively, and the maximum transmit power is $P_{\max} = 0.1$ W. The flight period, ISAC frame length, and time slot length are denoted by $T = 80$ s, $T_L = 20$ s, and $\delta_t = 0.25 $ s, respectively. The minimum achievable rate requirement is set as $R^{th}_k = 0.25$ bps/Hz.

We compare our proposed mechanism to two benchmarks:
\begin{itemize}
	\item {\bf{Straight flight (SF)}}: The UAV flies from the initial location ${\bm{q}}_I$ to the final location ${\bm{q}}_F$ along the straight line at the constant speed of $\frac{|{\bm{q}}_F - {\bm{q}}_I|}{T}$. 
	\item {\bf{Fly-Hover-Fly (FHF)}}: The UAV flies straightly at its maximum speed from the initial location to the optimized location obtained via solving the following problem:
	\begin{alignat}{2}
		\label{P3.1.6}
		& \begin{array}{*{20}{c}}
			\mathop {\max }\limits_{ {{\bm{w}}}_c, {{\bm{q}}}, {\bm{A}}, {\bm{C}}} \quad \mathop \sum\nolimits_{k = 1}^K  R_k
		\end{array} & \\ 
		\mbox{s.t.}\quad
		& ({\ref{P1}}a)-({\ref{P1}}d).  & \nonumber
	\end{alignat} 
	After hovering at this optimized location, the UAV flies straightly at its maximum speed to the final location.
\end{itemize}
Except for the UAV trajectory, the corresponding beamforming, user association, and sensing time slots during the flight period of these two benchmarks are obtained by {\bf{Algorithm} \ref{PenaltyBasedAlgorithm}} without updating the UAV trajectory. 

\begin{figure*}[h]
	\centering
	\setlength{\abovecaptionskip}{0.cm}
	
	\subfigure[$\Gamma^{th} = 0$.]
	{	
		\label{figure7a}
		\includegraphics[width=7.1cm]{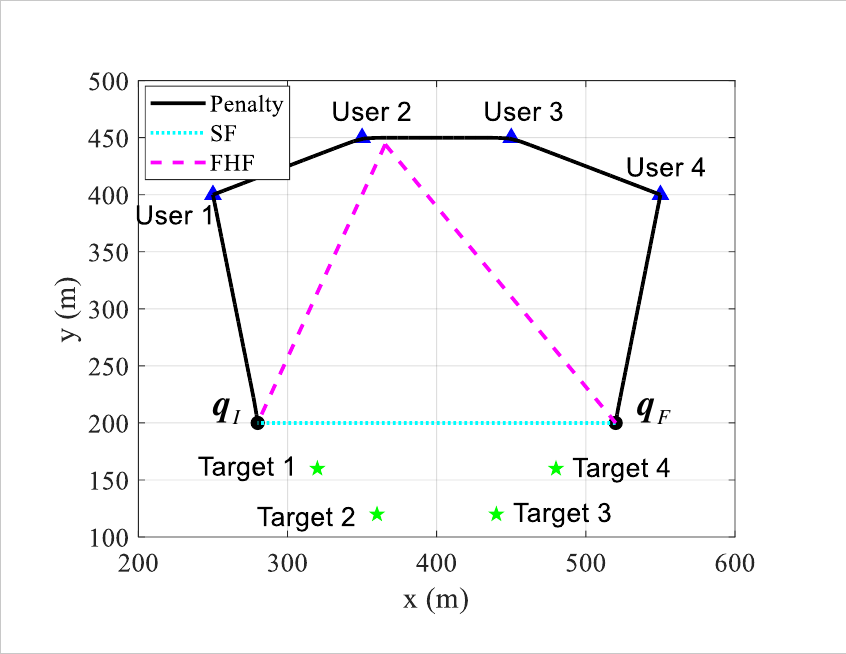}
	}
	\hspace{15mm}
	\subfigure[$\Gamma^{th} = 2 \times 10^{-5}$.]
	{	
		\label{figure7b}
		\includegraphics[width=7.1cm]{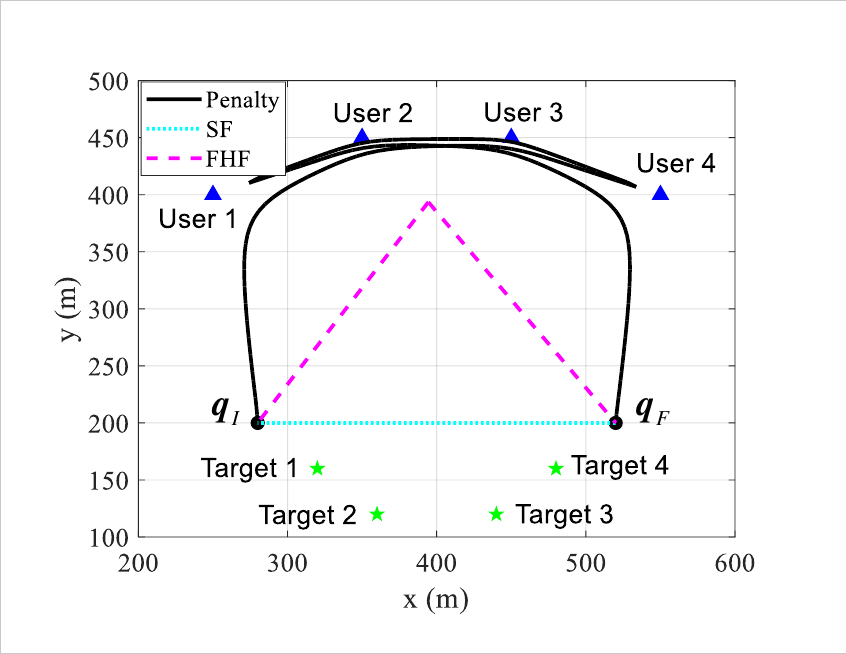}
	}
	\subfigure[$\Gamma^{th} = 6 \times 10^{-5}$.]
	{	
		\label{figure7c}
		\includegraphics[width=7.1cm]{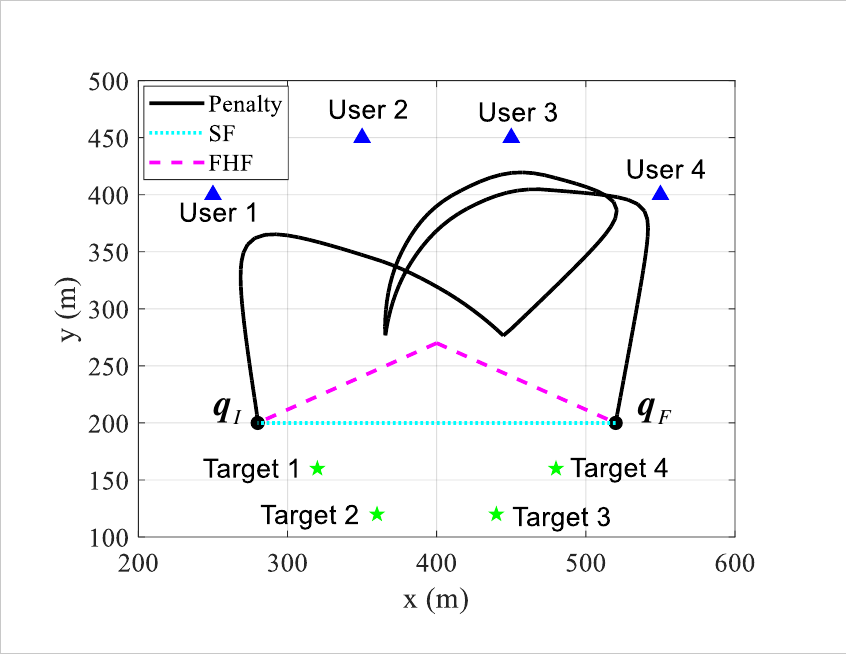}
	}
	\hspace{15mm}
	\subfigure[$\Gamma^{th} = 12 \times 10^{-5}$.]
	{	
		\label{figure7d}
		\includegraphics[width=7.1cm]{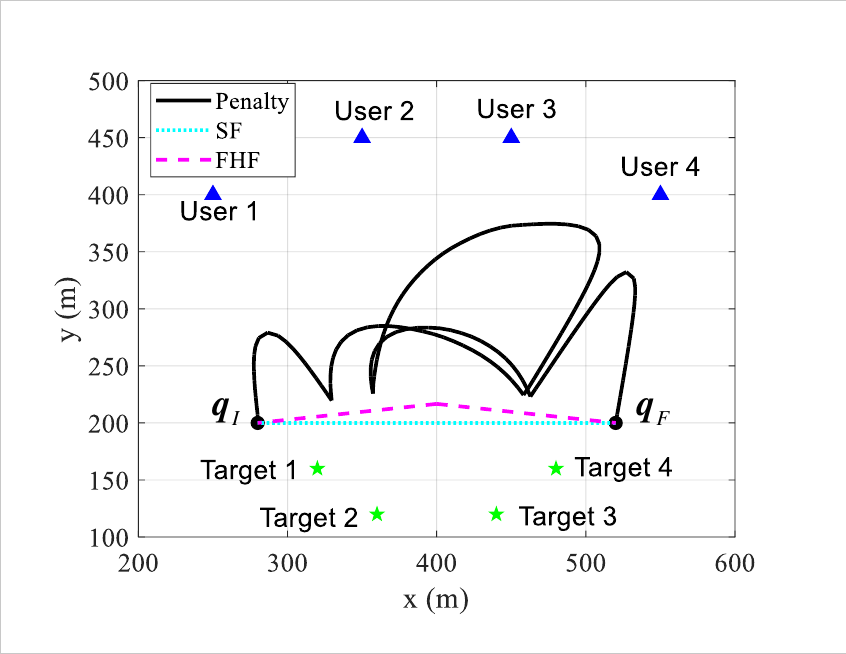}
	}
	\caption{UAV trajectories comparisons among the proposed penalty-based algorithm and benchmarks under different $\Gamma^{th}$ ($T = 40$ s and $T_L = 20$ s).}
	\label{figure7}
\end{figure*}

\subsection{Comparison Versus Sensing Power Requirement}

In Figs.~\ref{figure7} and \ref{figure8}, the UAV trajectories and the maximum achievable rate are illustrated respectively under different beam pattern gain thresholds $\Gamma^{th}$ for our proposed penalty-based algorithm (Solving problem (P1)) and benchmark schemes. Specifically, it can be observed from Fig.~\ref{figure7} that as the beam pattern gain threshold $\Gamma^{th}$ increases, the UAV's trajectory shrinks gradually from a relatively larger arc toward users to several smaller arcs between the targets and the users; the closest distance from the UAV to the users also increases since the UAV needs to perform sensing tasks at a location closer to the targets. In particular, when $\Gamma^{th} = 0$, i.e., no beam pattern gain constraint is considered as in \cite{Wu2018Common}, the UAV sequentially visits and stays above each of the users by maximally exploiting its mobility; while when $\Gamma^{th} = 12 \times 10^{-5}$, the UAV flies within a smaller region close to the targets due to the higher sensing power requirement. Notice that in this setup, the closer the UAV flies to the targets, the farther it is away from the communication users inevitably. As a result, satisfying the beam pattern gain requirements of the targets will consume more transmit power and thus becomes the bottleneck for improving the maximum achievable rate of the system. Such a situation will become worse when the beam pattern gain and/or the distance between the users and the targets becomes larger.

\begin{figure}[t]
	\centering
	\setlength{\abovecaptionskip}{0.cm}
	\includegraphics[width=7.1cm]{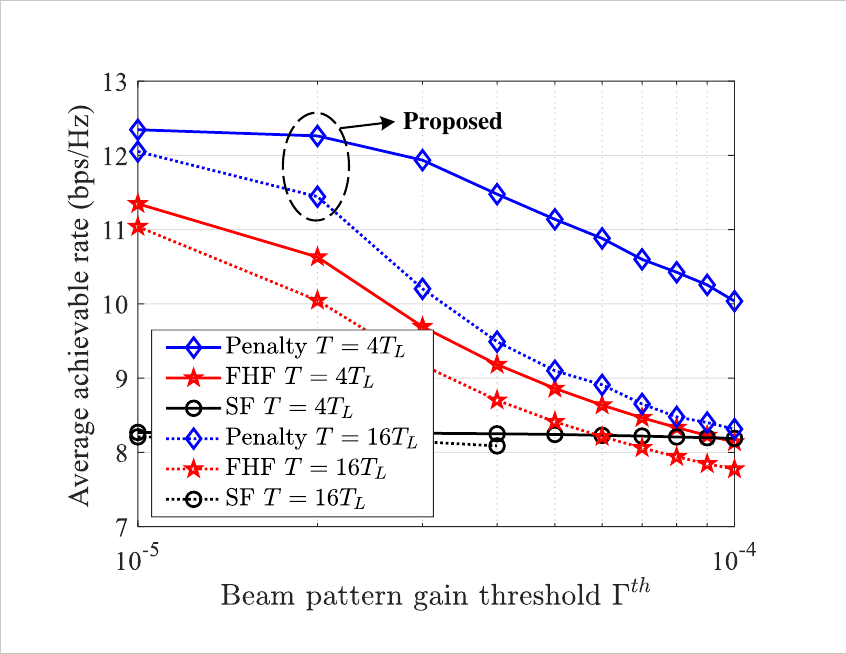}
	\caption{Achievable rate versus beam pattern gain threshold.}
	\label{figure8}
\end{figure}

The effect of the beam pattern gain constraints on the maximum achievable rate is shown in Fig.~\ref{figure8}. It is observed from Fig.~\ref{figure8} that the achievable rate gradually decreases as the beam pattern gain threshold $\Gamma^{th}$ increases. Also, the achievable rate gain achieved by our proposed scheme over the "SF" scheme increases as the sensing power requirement decreases, since the UAV's trajectory can be optimized in a larger feasible region for communication performance improvement. When the beam pattern gain threshold $\Gamma^{th}$ is larger than $4 \times 10^{-5}$, the "SF" scheme will become infeasible under the high-frequency sensing requirement, since the QoS constraints of users and the beam pattern gain constraints of targets cannot be satisfied without optimizing UAV trajectory. Moreover, the achievable rate of our proposed scheme achieves significant improvement as compared to the "FHF" scheme under lower sensing frequency, since the low-frequency sensing scenario shares more communication-only time slots in each ISAC frame for improving communication performance.  

\subsection{Comparison Versus Sensing Frequency}
\begin{figure*}[t]
	\centering
	\setlength{\abovecaptionskip}{0.cm}
	\subfigure[$T_L = 40 s$.]
	{	
		\label{figure9a}
		\includegraphics[width=7.1cm]{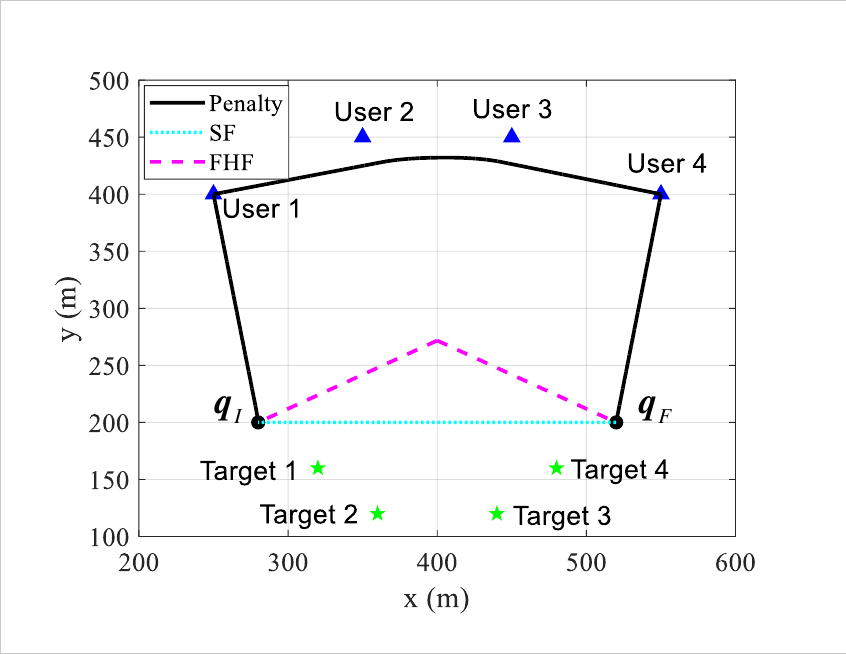}
	}
	\hspace{15mm}
	\subfigure[$T_L = 20 s$.]
	{	
		\label{figure9b}
		\includegraphics[width=7.1cm]{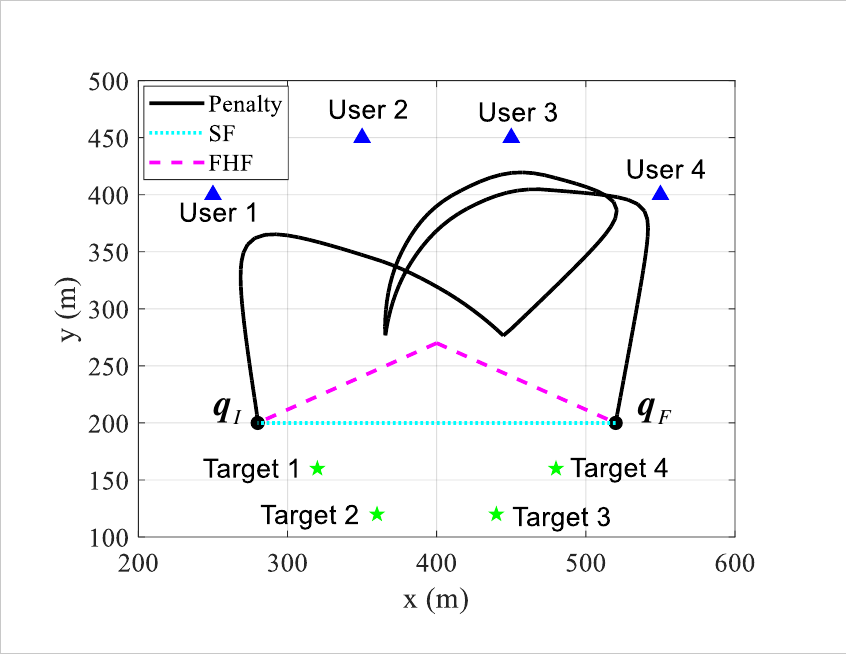}
	}
	\subfigure[$T_L = 10 s$.]
	{	
		\label{figure9c}
		\includegraphics[width=7.1cm]{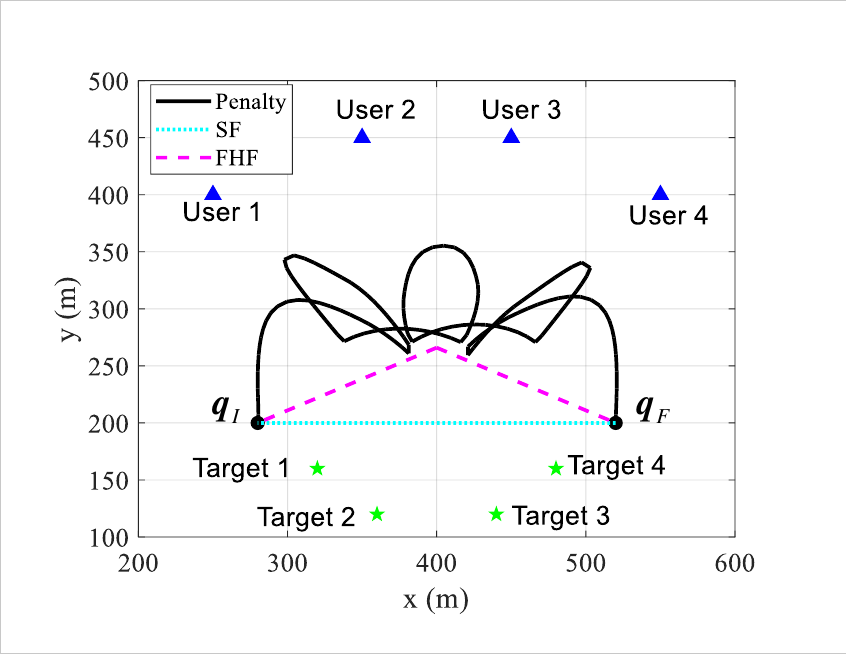}
	}
	\hspace{15mm}
	\subfigure[$T_L = 5 s$.]
	{	
		\label{figure9d}
		\includegraphics[width=7.1cm]{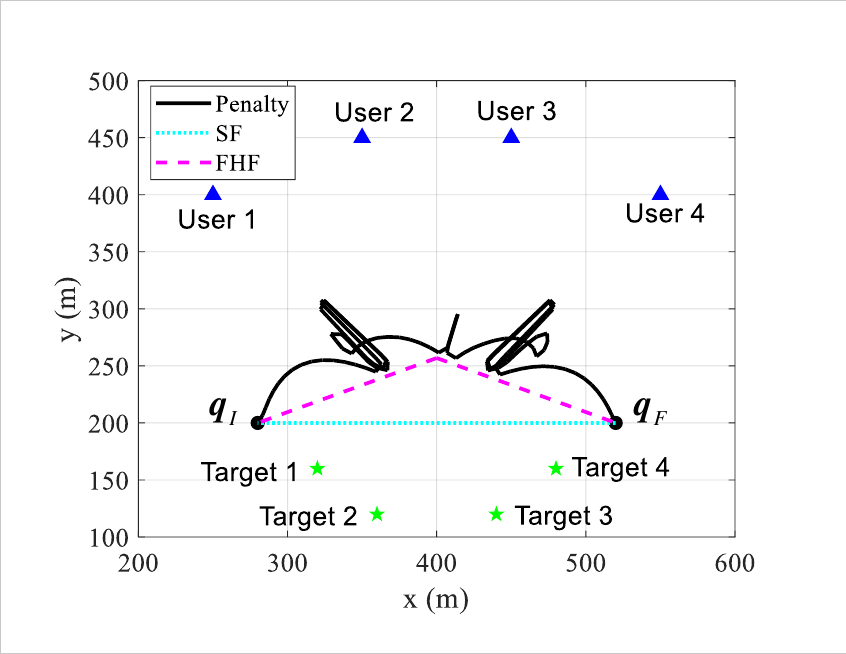}
	}	
	\caption{UAV trajectories comparisons among the proposed penalty-based algorithm and benchmarks under different sensing frequency (defined as ${1 \mathord{\left/{\vphantom {1 {{T_L }}}} \right.\kern-\nulldelimiterspace} {{T_L}}}$).}
	\label{figure9}
\end{figure*}
In Figs.~\ref{figure9} and \ref{figure10}, we show the UAV trajectories and the maximum achievable rate under different sensing frequency (defined as ${1 \mathord{\left/{\vphantom {1 {{T_L }}}} \right.\kern-\nulldelimiterspace} {{T_L}}}$) for our proposed penalty-based algorithm and benchmark schemes. Specifically, it can be observed from Fig.~\ref{figure9} that as the sensing frequency increases, the UAV's trajectory shares more turn-backs between the targets and the users since there exist more ISAC frames within a given flight period $T = 40$ s. In particular, when $T = T_L$, i.e., there is only one sensing time for each target, the UAV can almost fly above each of the users to achieve better air-to-ground channels between the UAV and each user; when $T = 8 T_L$, the UAV trajectory consists of multiple almost overlapping trajectory segments between the targets and one certain user. Generally speaking, as the sensing frequency increases, the UAV trajectory tends to be more restricted to avoid getting too far away from any of the targets. 

Fig.~\ref{figure10} shows the performance comparison among sensing power requirement, sensing frequency, and achievable rate. Specifically, as the sensing frequency increases, the achievable rate of all the considered mechanisms decreases, which validates the analysis in Proposition \ref{IncreaseNL}. Also, the achievable rate of our proposed algorithm under a higher beam pattern gain threshold ${{\Gamma}^{th}}$ degrades faster as compared to that under a lower threshold. The main reason is that a higher beam pattern gain threshold forces the UAV to perform sensing tasks at a location closer to the target, thereby resulting in increasing path loss within the communication-only duration. Furthermore, it is observed from Fig.~\ref{figure10} that the achievable rate gain achieved by our proposed scheme over the "FHF" and "SF" schemes increases as the sensing frequency decreases, as the UAV has more non-sensing time to adjust its trajectory for communication performance improvement. 

\begin{figure}[t]
	\centering
	\setlength{\abovecaptionskip}{0.cm}
	\includegraphics[width=7.1cm]{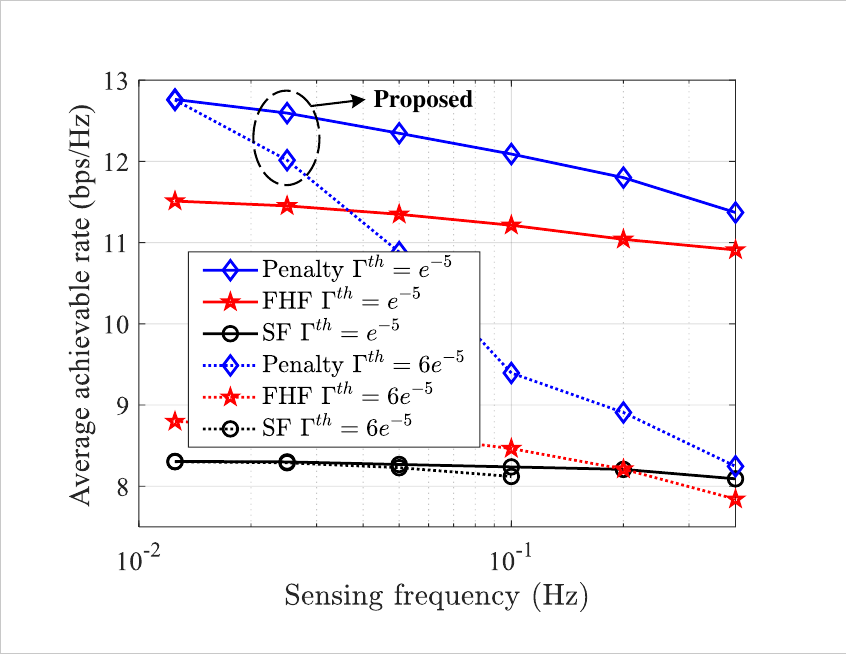}
	\caption{Achievable rate versus sensing frequency requirement.}
	\label{figure10}
\end{figure}

\subsection{User-Target Association and Beam Pattern}
Next, the user association and target selection at the sensing time slots are shown in Fig.~\ref{figure11}, where $T = T_L = 40$ s, and $\Gamma^{th} = 10^{-3}$. The UAV's flight speed is illustrated in Fig.~\ref{figure11a}, where the user association and target selection are represented by blue and green dashed lines, respectively. Besides, it can be seen from Fig.~\ref{figure11a} that the UAV tends to provide the communication service for the user which is closer to the associated target. The beam pattern gains in space at two selected sensing time slots are shown in Figs.~\ref{figure11b} and \ref{figure11c}, where the beams are mainly concentrated in the direction of the selected target's location and the associated user's position.
\begin{figure*}[htpb]
	\centering
	\setlength{\abovecaptionskip}{0.cm}
	\subfigure[UAV trajectory with its speed.]
	{	
		\label{figure11a}
		\includegraphics[width=5.3cm]{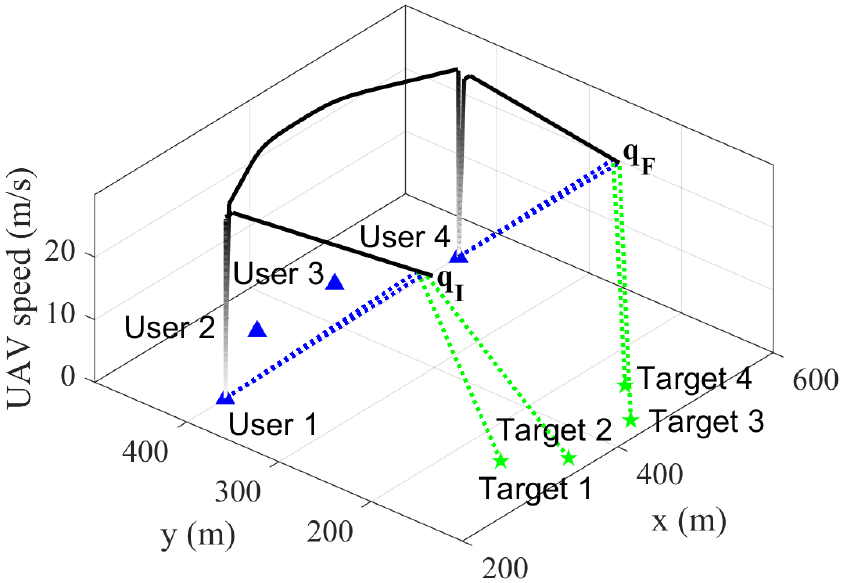}
	}
	\subfigure[Beam pattern at $n$ = 2.]
	{	
		\label{figure11b}
		\includegraphics[width=5cm]{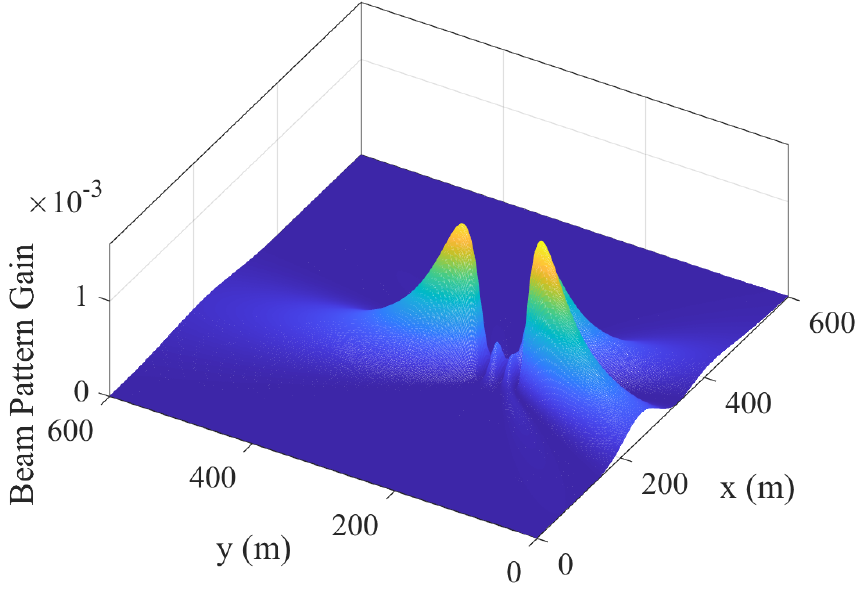}
	}
	\subfigure[Beam pattern at $n$ = 159.]
	{	
		\label{figure11c}
		\includegraphics[width=5cm]{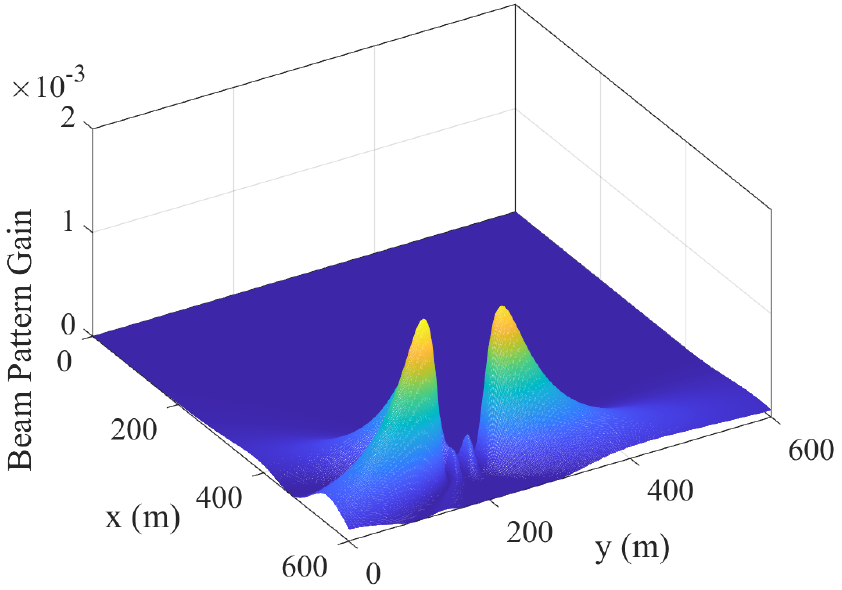}
	}
	\caption{UAV trajectory and its corresponding beam pattern gain at sensing time slots.}
	\label{figure11}
\end{figure*}

\subsection{Lower Bound's Gap and Low Complexity Method}
\label{LowerBoundDifference}
Moreover, since problem (P2) is one approximation of problem (P1), we substitute the optimized solution obtained by {\bf{Algorithm} \ref{PenaltyBasedAlgorithm}} back into the objective function of problem (P1) to obtain the actual achievable user rate, as shown in Fig.~\ref{figure15} for comparison. Specifically, the difference of average achievable rate during sensing time between original objective value $R^{ISAC}_{k,j}$ and approximate objective value $\underline{R}^{ISAC}_{k,j}$ will decrease as the number of antennas increases, where $M_x = M_y$. In particular, the average achievable rate of the original objective value is almost approximated to the objective value (less than 1$\%$) when the number of antennas $M$ is larger than 16, which justifies the accuracy of our derived lower bound in Lemma \ref{OptimalBeamformingLower}.

\begin{figure*}[t]
	\centering
	\setlength{\abovecaptionskip}{0.cm}
	\subfigure[Comparisons under different number of antennas.]
	{	
		\label{figure15}
		\includegraphics[width=7.1cm]{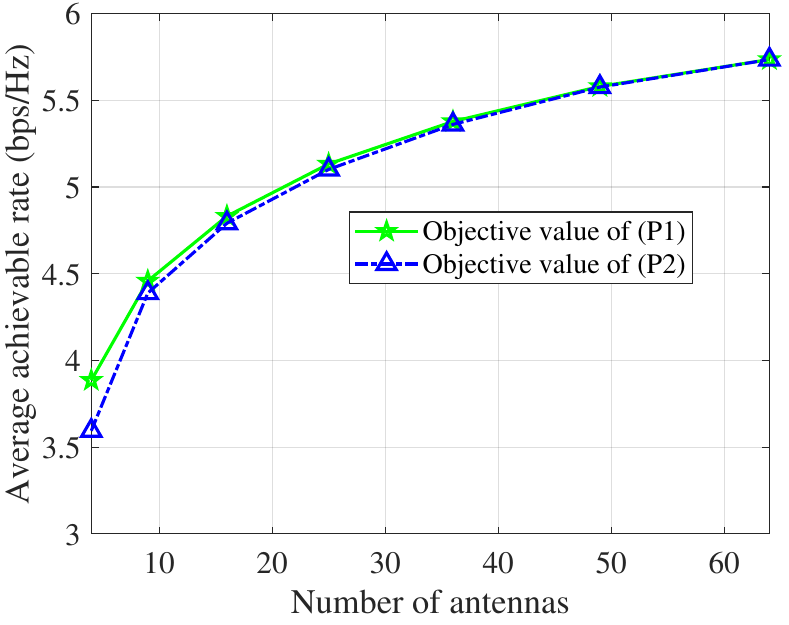}
	} \hspace{15mm}
	\subfigure[Comparisons versus different flight periods.]
	{	
		\label{figure16}
		\includegraphics[width=7.1cm]{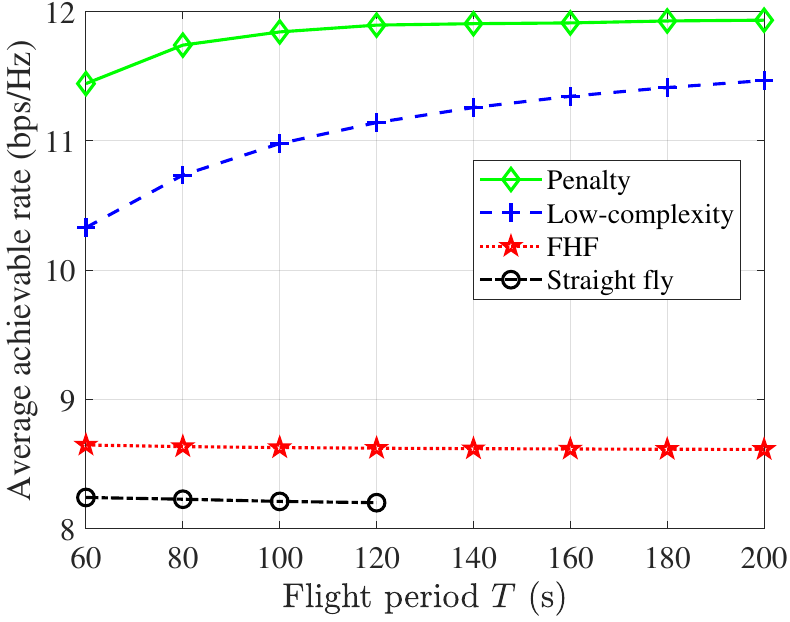}
	}
	\caption{Comparisons for proposed penalty-based algorithm and low-complexity algorithm.}
	\label{figure156}
\end{figure*}

Moreover, we show the communication performance difference between the proposed penalty-based algorithm (refers to {\bf{Algorithm} 1}) and the proposed low-complexity algorithm (the proposed algorithm in Section IV-B) under different flight periods $T$ in Fig.~\ref{figure16}. The low-complexity algorithm can achieve a higher gain over the two benchmarks as the flight period increases. Interestingly, the achievable rate gain achieved by the penalty-based algorithm over the low-complexity algorithm will decrease as the flight period increases. In particular, for the proposed low-complexity algorithm, there is only no more than 5$\%$ performance loss as compared to the proposed penalty-based algorithm when the flight period is larger than 200 s. This is due to the derived structural characteristics of the optimal solutions among different ISAC frames. Specifically, for large flight periods, the flight time from the initial location or that to the final location accounts for a smaller proportion of the entire flight period $T$, and the corresponding communication rate is approximate to that without the location constraints.

\subsection{Pathloss Factor For Sensing}
\label{DifferentPathloss}
The effect of different pathloss factors for sensing power, i.e., the exponent of the distance in (\ref{P1}a), is further evaluated in Fig.~\ref{figure13}. Fig.~\ref{figure13a} shows that under the pathloss with the fourth power of the distance, the UAV trajectory shares several turn-backs between the targets and the users under $T=2T_L$ and $\Gamma^{th} = 8 \times 10^{-9}$, which is similar to that in Fig.~\ref{figure7c} but with a much lower beam pattern gain threshold. It can be seen that the pathloss factor mainly affects the distance between the UAV and the target when performing sensing tasks, and has little effect on the overall trajectory trend. Fig.~\ref{figure13b} shows that under high sensing frequency, the achievable rate also decreases in a similar trend with that in Fig.~\ref{figure8} as the beam pattern gain constraints increases; even under low sensing frequency, the achievable rate decreases faster since the UAV needs to perform sensing tasks at a location closer to the targets under the path loss related to the fourth power of distance.
\begin{figure*}[t]
	\centering
	\setlength{\abovecaptionskip}{0.cm}
	\subfigure[UAV trajectory under 4-exponent pathloss factor.]
	{	
		\label{figure13a}
		\includegraphics[width=7.1cm]{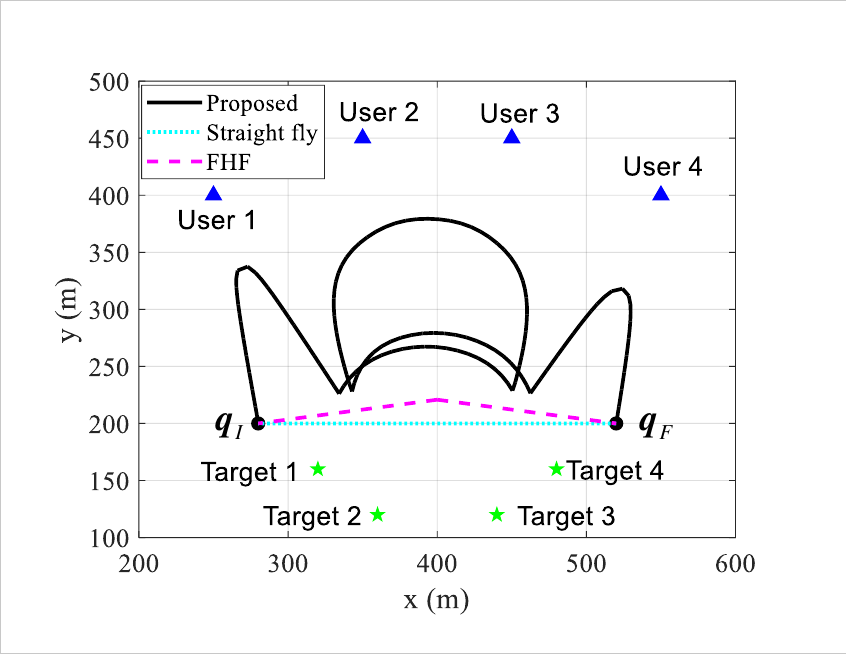}
	} \hspace{15mm}
	\subfigure[Comparisons versus beam pattern gain constraints.]
	{	
		\label{figure13b}
		\includegraphics[width=7.1cm]{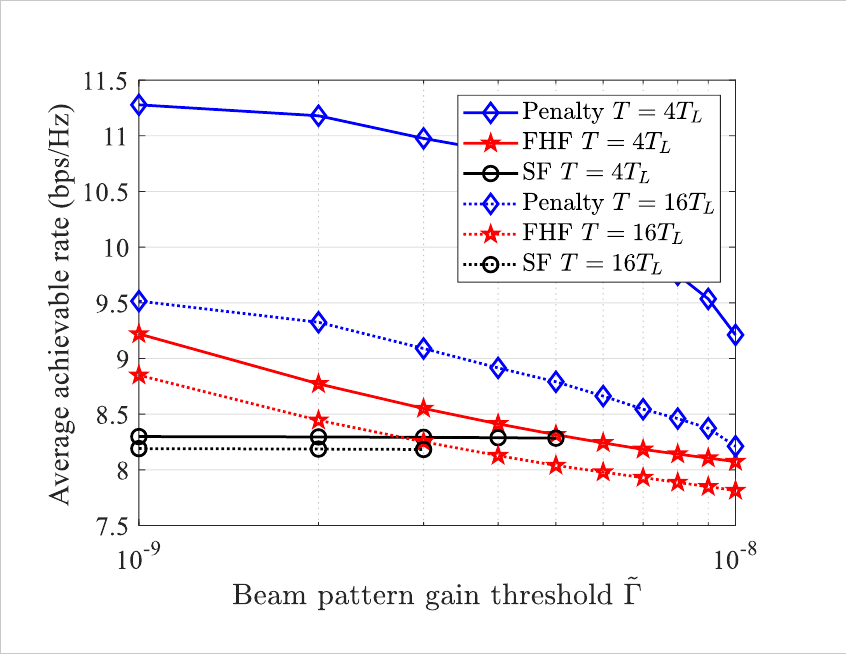}
	}
	\caption{The UAV trajectory and the achievable rate comparison under the pathloss factor with power of 4.}
	\label{figure13}
\end{figure*}

\section{Conclusions and Future Works}
\label{Conclusion}
In this paper, we investigated a new type of UAV-enabled periodic ISAC system. Specifically, the beamforming, user association, sensing time selection, and UAV trajectory were jointly optimized to maximize the sum achievable rate. The closed-form optimal beamforming vector was derived to significantly reduce the complexity of beamforming design, and a lower bound of the achievable rate was presented to facilitate UAV trajectory design. By ignoring the initial and final location constraints, a novel symmetric structure of the optimal solutions among adjacent frames was identified to reveal a fundamental trade-off between sensing frequency and communication capacity. Based on this, a low-complexity method was presented based on our derived structural characteristics. The numerical results validated the efficiency of our design over the benchmark schemes and also confirmed the benefits of the periodic ISAC framework. The more general cases considering the effects caused by imperfectly compensated Doppler for multi-UAV ISAC scenarios are worthwhile future works. In addition, the sensing-assisted communication problems considering the sensing gain and clutter interference will be further investigated in our future work.

\section*{Appendix A: \textsc{Proof of Proposition 1}}
\label{ProveP2}

For $ \frac{{M{P_{\max }}\cos^2 { \varphi _{k,j}}}}{d({\bm{q}}[n], {\bm{v}}_j)^2} \ge  {\Gamma ^{th}}$, We can readily derive that the beam pattern gain at target will be no less than the threshold $\Gamma^{th}$ if the optimal beamforming vector is $\sqrt{P_{\max}}\frac{{\bm{h}}_{c,k}}{\|{\bm{h}}_{c,k}\|}$. In the following, we prove that for $ \frac{{{M{P_{\max }}\cos^2 { \varphi _{k,j}}}}}{d({\bm{q}}[n], {\bm{v}}_j)^2} <  {\Gamma ^{th}}$, the optimal beamforming vector equals to $\frac{1}{{{\lambda _1}}}( {\sqrt {{\beta_{c,k}}}{{\bm{h}}_{c,k}}  + {\lambda _2} \sqrt {{\Gamma^{th}}} {{\bm{h}}_{r,j}}  {e^{ \jmath {\varphi _{k,j}}}}} )$.

First, it can be easily shown that constraint (\ref{P1.1}b) is met with equality for the optimal solution since otherwise $\|{\bm{w}}_c\|$ can be always increased to improve the objective value until (\ref{P1.1}b) becomes active. Hence, constraint (\ref{P1.1}b) can be rewritten as $\|{{\bm{w}}_c}\|^2 = P_{\max}$. Hence, the corresponding Lagrangian function of (\ref{P1.1}) is given by 
\begin{equation}
	\begin{aligned}
	L({\bm{w}}_c,{\lambda _1},{\lambda _2}) \!=\! &   \!-\! {{\bm{w}}_c^H}{{\bm{h}}_{c,k}}{\bm{h}}_{c,k}^H{\bm{w}}_c  \!+\! {\lambda _1}( {{{\| {\bm{w}}_c \|}^2} - {P_{\max }}} ) \\
	&+ {\lambda _2}\left( {{\Gamma^{th}} - {{\bm{w}}_c^H}{{\bm{h}}_{r,j}}{\bm{h}}_{r,j}^H{\bm{w}}_c} \right).		
	\end{aligned}
\end{equation}
We can construct the Karush-Kuhn-Tucker (KKT) conditions for the optimal solution at a feasible point as follows:
\begin{equation}\label{KKTcondition}
	{{\nabla  L({\bm{w}}_c,{\lambda _1},{\lambda _2})}} =  - {{\bm{h}}_{c,k}}{\bm{h}}_{c,k}^H{\bm{w}}_c + {\lambda _1}{\bm{w}}_c - {\lambda _2}{{\bm{h}}_{r,j}}{\bm{h}}_{r,j}^H{\bm{w}}_c = 0,
\end{equation}
\begin{equation}\label{OptimalCondition2}
	\lambda_2\left(\Gamma^{th}  -  {{\bm{w}}_c^H} {\bm{h}}_{r,j} {\bm{h}}_{r,j}^H{\bm{w}}_c\right) = 0.
\end{equation}
From (\ref{KKTcondition}), it can be shown that 
\begin{equation}\label{KKTcondiiton1}
	{{{\bm{h}}_{c,k}}{\bm{h}}_{c,k}^H{\bm{w}}_c{{ + }}{\lambda _2}{{\bm{h}}_{r,j}}{\bm{h}}_{r,j}^H{\bm{w}}_c = {\lambda _1}{\bm{w}}_c}.
\end{equation}
Multiplying both sides of equation (\ref{KKTcondiiton1}) with ${\bm{w}}_c$ leads to
\begin{equation}
	{{\bm{w}}_c^H}{{\bm{h}}_{c,k}}{\bm{h}}_{c,k}^H{\bm{w}}_c{{ + }}{\lambda _2}{{\bm{w}}_c^H}{{\bm{h}}_{r,j}}{\bm{h}}_{r,j}^H{\bm{w}}_c = {\lambda _1}{{\bm{w}}_c^H}{\bm{w}}_c = \lambda_1 P_{\max}.
\end{equation}
Let ${{ {\bm{h}}_{c,k}^H{\bm{w}}_c} } =  \sqrt{{\beta_{c,k}}}e^{j\varphi_{c,k}}$, ${\bm{h}}^H_{r,j}{\bm{w}}_c =  \sqrt{\beta_{r,j}}e^{j\varphi_{r,j}}$, it follows that
\begin{equation}\label{SolvingEquation3}
	{\beta_{c,k}} + {\lambda _2}{\beta_{r,j}} = {\lambda _1}{P_{\max }}.
\end{equation}
Define ${\bm{H}} = [{\bm{h}}_{c,k}, {\bm{h}}_{r,j}]$, by multiplying both sides of equation (\ref{KKTcondiiton1}) with $\left({{\bm{H}}^H{\bm{H}}}\right)^{-1}{\bm{H}}^H$, equation (\ref{KKTcondiiton1}) becomes 
\begin{equation}\label{EquationBetaLamda}
	\begin{aligned}
		&\left[ {\begin{array}{*{20}{c}}
				\sqrt {{\beta_{c,k}}} {e^{j{\varphi_{c,k}}}}\\
				{{\lambda _2}{\sqrt {{\beta_{r,j}}}} {e^{j{\varphi_{r,j}}}}}
		\end{array}} \right] = {\lambda _1}{\left( {{{\bm{H}}^{H}}{\bm{H}}} \right)^{ - 1}}{{\bm{H}}^{H}}{\bm{w}}_c \\
		=& \frac{{{\lambda _1}}}{V_{k,j}}\left[ {\begin{array}{*{20}{c}}
				{{{| {{{\bm{h}}_{r,j}}} |}^2}}&{ - {\bm{h}}_{c,k}^H{{\bm{h}}_{r,j}}}\\
				{ - {{( {{\bm{h}}_{c,k}^H{{\bm{h}}_{r,j}}} )}^H}}&{{{| {{\bm{h}}_{c,k}^H} |}^2}}
		\end{array}} \right]\left[ {\begin{array}{*{20}{c}}
				{\sqrt {{\beta_{c,k}}} {e^{j{\varphi_{c,k}}}}}\\
				{\sqrt {{\beta_{r,j}}}} {e^{j{\varphi_{r,j}}}}
		\end{array}} \right].
	\end{aligned}
\end{equation}
In (\ref{EquationBetaLamda}), $V_{k,j} = {{\left\| {\bm{h}}_{c,k}^H \right\|}^2}{{\left\| {\bm{h}}_{r,j} \right\|}^2} - {\left| {{\bm{h}}_{c,k}^H{{\bm{h}}_{r,j}}} \right|}^2 \ne 0$, otherwise MRT is the optimal beamforming. If $\lambda_2 = 0$, it follows that $\sqrt {{\beta _{r,j}}}  = \frac{{\left\| {{\bm{h}}_{c,k}^H{\bm{h}}_{r,j}} \right\|}}{{{{\left\| {{\bm{h}}_{c,k}^H} \right\|}^2}}}\sqrt {{\beta _{c,k}}} $ according to (\ref{EquationBetaLamda}). By plugging this condition into (\ref{SolvingEquation3}), $\beta_{c,k} = P_{\max} \|{\bm{h}}_{c,k}\|$, which holds if and only if ${\bm{w}^*} = \sqrt{P_{\max}} \frac{{\bm{h}}_{c,k}}{\|{\bm{h}}_{c,k}\|}$. When $\lambda_2 \ne 0$, the KKT condition in (\ref{OptimalCondition2}) can be written as  $ {{\bm{w}}_c^H} {\bm{h}}_{r,j} {\bm{h}}_{r,j}^H{\bm{w}}_c = \beta_ {r,j} = \Gamma^{th}$. Since $\lambda_1$ and $\lambda_2$ are real-valued, equation (\ref{EquationBetaLamda}) can be rewritten as
\begin{equation}\label{SolvingEquation1}
	{\left[ {\begin{array}{*{20}{c}}
				{\sqrt {{\beta_{c,k}}} }\\
				{{\lambda _2}\sqrt {{\Gamma ^{th}}} }
		\end{array}} \right] = \frac{{{\lambda _1}}}{V_{k,j}}\left[ {\begin{array}{*{20}{c}}
				{{{\| {{{\bm{h}}_{r,j}}} \|}^2}\sqrt {{\beta_{c,k}}}  - | {{\bm{h}}_{c,k}^H{{\bm{h}}_{r,j}}} |\sqrt {{\Gamma ^{th}}} }\\
				{{{\| {{\bm{h}}_{c,k}^H} \|}^2}\sqrt {{\Gamma ^{th}}}  - | {{\bm{h}}_{c,k}^H{{\bm{h}}_{r,j}}} |\sqrt {{\beta_{c,k}}} }
		\end{array}} \right]}
\end{equation}
and
\begin{equation}\label{SolvingEquation2}
	{\left[ {\begin{array}{*{20}{c}}
				{\sqrt {{\beta_{c,k}}} }\\
				{{\lambda _2}\sqrt {{\Gamma ^{th}}} }
		\end{array}} \right] = \frac{{{\lambda _1}}}{V_{k,j}}\left[ {\begin{array}{*{20}{c}}
				{{{\| {{{\bm{h}}_{r,j}}} \|}^2}\sqrt {{\beta_{c,k}}}  + | {{\bm{h}}_{c,k}^H{{\bm{h}}_{r,j}}} |\sqrt {{\Gamma ^{th}}} }\\
				{{{\| {{\bm{h}}_{c,k}^H} \|}^2}\sqrt {{\Gamma ^{th}}}  + | {{\bm{h}}_{c,k}^H{{\bm{h}}_{r,j}}} |\sqrt {{\beta_{c,k}}} }
		\end{array}} \right]},
\end{equation}
when $\varphi_{r,j} - \varphi_{c,k} = -  {\varphi _{k,j}} + 2n \pi$ and $\varphi_{r,j} - \varphi_{c,k} = -  {\varphi _{k,j}} + (2n+1) \pi, n \in {\mathbb{Z}}$, respectively. By plugging (\ref{SolvingEquation1}) or (\ref{SolvingEquation2}) into (\ref{SolvingEquation3}), then $\beta_{c,k}$ can be expressed as
\begin{equation}\label{BetaLarge}
	{\beta^+_{c,k}}{{ = }}\frac{{\left\| {{\bm{h}}_{c,k}^H} \right\|^2}}{{\left\| {{{\bm{h}}_{r,j}}} \right\|^2}}{\left( {\sqrt {{\Gamma^{th}}}\cos  {\varphi _{k,j}} {{ + }}{P_j}\sin  {\varphi _{k,j}} } \right)^2}
\end{equation}
or
\begin{equation}
	{\beta^-_{c,k}}{{ = }}\frac{{\left\| {{\bm{h}}_{c,k}^H} \right\|^2}}{{\left\| {{{\bm{h}}_{r,j}}} \right\|^2}}{\left( {\sqrt {{\Gamma^{th}}}\cos  {\varphi _{k,j}} {{ - }}P_j\sin  {\varphi _{k,j}} } \right)^2},
\end{equation}
where $P_j = \sqrt {{P_{\max }}{{\left\| {{{\bm{h}}_{r,j}}} \right\|}^2} - {\Gamma^{th}}}$ and ${ \varphi _{k,j}} = \arccos \frac{|{\bm{h}}_{c,k}^H {\bm{h}}_{r,j}|}{\| {\bm{h}}_{c,k}^H \| \| {\bm{h}}_{r,j} \|}$. Since $ \beta^+_{c,k} >  \beta^-_{c,k}$, the optimal solution to problem in (\ref{P1.1}) can be obtained when $\beta_{c,k} = \beta_{c,k}^+$. Then, by plugging (\ref{BetaLarge}) into (\ref{SolvingEquation1}), we have ${\lambda _1^*} = \frac{\Upsilon {\left\| {{\bm{h}}_{c,k}^H} \right\|}^2{\sin { \varphi _{k,j}} }}{{\sqrt {{P_{\max }}{{\left\| {{{\bm{h}}_{r,j}}} \right\|}^2} - {\Gamma ^{th}}} }}$ and ${\lambda _2^*} = \frac{{\Upsilon  {{{\left\| {{\bm{h}}_{c,k}^H} \right\|}^2}\sqrt {{\Gamma ^{th}}}  - \Upsilon^2 \left\| {{\bm{h}}_{c,k}^H} \right\|\left\| {{{\bm{h}}_{r,j}}} \right\|\cos { \varphi _{k,j}} } }}{{{{\left\| {{{\bm{h}}_{r,j}}} \right\|}^2}\sqrt {{P_{\max }}{{\left\| {{{\bm{h}}_{r,j}}} \right\|}^2}{\Gamma ^{th}} - {{\left( {{\Gamma ^{th}}} \right)}^2}}\sin { \varphi _{k,j}} }}$, where $\Upsilon =  {\sqrt {{\Gamma^{th}}}}\cos  {\varphi _{k,j}} +  {\sqrt {{P_{\max }}{{\left\| {{{\bm{h}}_{r,j}}} \right\|}^2} - {\Gamma^{th}}}}\sin  {\varphi _{k,j}}$.
Hence, the optimal beamforming can be expressed as
\begin{equation}
	{\bm{w}}_c^* = \frac{1}{\lambda_1^*}\left( \sqrt{\beta_{c,k}^+}{\bm{h}}_{c,k} + \lambda_2^* \sqrt{\Gamma^{th}}{\bm{h}}_{r,j} e^{-\jmath  {\varphi _{k,j}}} \right).
\end{equation}
By combining the above results above, Proposition {\ref{OptimalBearfoming}} is finally proved.

\section*{Appendix B: \textsc{Proof of Lemma \ref{InftyAntenaNumber}}}

Let $\Delta \Omega = {\Omega ({{\bm{q}}[n]},{\bm{v}}_j) - \Omega( {{\bm{q}}[n]},{\bm{u}}_k)}$ and $\Delta \Phi = {\Phi ({{\bm{q}}[n]},{\bm{v}}_j) - \Phi ({{\bm{q}}[n]},{\bm{u}}_k)}$. When $\Delta \Omega = 0$ and $\Delta \Phi = 0$, i.e., ${\bm{u}}_k = {\bm{v}}_j$, then $\gamma_0^* = \gamma_0 \frac{{ {M{P_{\max }} } }}{{d({\bm{q}}[n],{\bm{u}}_k)^2}}$. When $\Delta \Omega \ne 0$, and $\Delta \Phi \ne 0$, $|\cos { \varphi _{k,j}}|$ can be recast as shown in (\ref{CosValue}). 
\begin{figure*}[h]
	\begin{equation}\label{CosValue}
		\begin{aligned}
			|\cos { \varphi _{k,j}}| =& \frac{1}{M} \left| {\sum\nolimits_{{m_x} = 1}^{{M_x}} {{e^{\jmath \pi {m_x}\left( {\Phi ({{\bm{q}}[n]},{\bm{v}}_j) - \Phi ({{\bm{q}}[n]},{\bm{u}}_k)} \right)}}\sum\nolimits_{{m_y} = 1}^{{M_y}} {{e^{\jmath \pi {m_y}\left( {\Omega ({{\bm{q}}[n]},{\bm{v}}_j) - \Omega ({{\bm{q}}[n]},{\bm{u}}_k)} \right)}}} } } \right| \\
			=& \frac{1}{M}  \left| {{e^{\jmath \pi M{\Delta \Omega}/2 - \jmath \pi {\Delta \Omega}/2}}\left( {\frac{{{e^{ - \jmath \pi M{\Delta \Omega}/2}} - {e^{\jmath \pi M{\Delta \Omega}/2}}}}{{{e^{ - \jmath \pi {\Delta \Omega}/2}} - {e^{\jmath \pi {\Delta \Omega}/2}}}}} \right)\sum\nolimits_{{m_x} = 1}^{{M_x}} {{e^{\jmath \pi {m_x}\left( {\Phi ({{\bm{q}}[n]},{\bm{v}}_j) - \Phi ({{\bm{q}}[n]},{\bm{u}}_k)} \right)}}} } \right| \\
			= & \left| {\frac{{\sin M_x  \Delta \Phi  \pi /2}}{M_x{ \sin  \Delta \Phi \pi /2}}}   \right|   \left|{\frac{{\sin M_y  \Delta \Omega  \pi /2}}{M_y{ \sin  \Delta \Omega \pi /2}}}\right| .
		\end{aligned}
	\end{equation}
\end{figure*}
When $M_x \to \infty$ or $M_y \to \infty$, $\left| {\frac{{\sin M_x  \Delta  \Phi \pi /2}}{{M_x \sin  \Delta \Phi \pi /2}}}   \right| \left|{\frac{{\sin M_y  \Delta \Omega  \pi /2}}{{M_y \sin  \Delta \Omega \pi /2}}}\right| = 0$, as i.e., $\cos { \varphi _{k,j}} = 0$. When $\Delta \Omega = 0$ or $\Delta \Phi = 0$, $\cos { \varphi _{k,j}}$ can be transformed into $\left| {\frac{{\sin M_x  \Delta \Phi  \pi /2}}{{M_x \sin  \Delta \Phi \pi /2}}}   \right|$ or $\left|{\frac{{\sin M_y  \Delta \Omega  \pi /2}}{{M_y \sin  \Delta \Omega \pi /2}}}\right|$, respectively. Then, $\cos { \varphi _{k,j}} = 0$ when $M_x \to \infty$ and $M_y \to \infty$. In this case, ${\gamma^*_{k}} = \gamma_0 \frac{{ {M{P_{\max }} - {\Gamma ^{th}}{d({\bm{q}}[n],{\bm{v}}_j)^2}} }}{{d({\bm{q}}[n],{\bm{u}}_k)^2}}$. Thus, (\ref{InftyMBeta}) holds. 

Accordingly, we can readily prove that the optimal horizontal coordinate should be within the line formed by ${\bm{u}}_k$ and ${\bm{v}}_j$, Then, the horizontal distance from UAV to user $k$ is denoted by $x$, and $x \ne 0$ if ${\bm{u}}_k \ne {\bm{v}}_j$. By taking the derivative of $x$ to ${\gamma^*_{k}}$, the following condition holds:
\begin{equation}\label{HoveringLocationEquation}
	{x^2} + \left( {\frac{{M{P_{\max }}}}{{{\Gamma ^{th}}d({\bm{u}}_k,{\bm{v}}_j)}} - d({\bm{u}}_k,{\bm{v}}_j)} \right)x - {H^2} = 0.
\end{equation}
Then, the optimal UAV location can be obtained by solving the equation in (\ref{HoveringLocationEquation}), i.e., $x = \frac{{\sqrt {{Z^2} + 4{H^2}}  - Z}}{2}$, where $Z = \frac{{M{P_{\max }}}}{{{\Gamma ^{th}}D_{k,j}}} - D_{k,j}$. As $\frac{x}{D_{k,j}} = \frac{d({\bm{q}}[n],{\bm{u}}_k)}{d({\bm{u}}_k,{\bm{v}}_j)}$, then the UAV location with maximum achievable rate ${\bm{q}}^*_{k,j} = {\bm{u}}_k + \frac{{\sqrt {{{ Z}^2} + 4{H^2}} - Z }}{{2D_{k,j}}}({{\bm{v}}_j} - {{\bm{u}}_k})$, and thus complete the proof.

\footnotesize  	
\bibliography{mybibfile}
\bibliographystyle{IEEEtran}

\end{document}